\documentclass{article}
\usepackage{graphicx}
\usepackage{hyperref}
\usepackage{lineno}
\usepackage[OT1]{fontenc}
\usepackage{amsthm,amsmath,amssymb,mathrsfs}
\usepackage[numbers]{natbib}
\usepackage{graphicx}
\usepackage{subfig}
\usepackage{titlesec}
\usepackage{rotating}
\usepackage{multirow}
\RequirePackage{enumitem} 
\RequirePackage{mathtools}
\usepackage[doublespacing]{setspace}
\usepackage{geometry}
\numberwithin{equation}{section}
\theoremstyle{plain}
\newtheorem{thm}{Theorem}[section]
\newtheorem{lemma}[thm]{Lemma}
\titlelabel{\thetitle.\quad}
\def\bibindent{1em}

\makeatletter
\renewcommand\@biblabel[1]{} 
\renewenvironment{thebibliography}[1]
{\section*{\refname}%
	\@mkboth{\MakeUppercase\refname}{\MakeUppercase\refname}%
	\list{\@biblabel{\@arabic\c@enumiv}}%
	{\settowidth\labelwidth{\@biblabel{}}%
		\leftmargin\labelwidth
		\advance\leftmargin15pt
		\advance\leftmargin\labelsep
		\setlength\itemindent{-10pt}
		\@openbib@code
		\usecounter{enumiv}%
		\let\p@enumiv\@empty
		\renewcommand\theenumiv{\@arabic\c@enumiv}}%
	\sloppy
	\clubpenalty4000
	\@clubpenalty \clubpenalty
	\widowpenalty4000%
	\sfcode`\.\@m}
{\def\@noitemerr
	{\@latex@warning{Empty `thebibliography' environment}}%
	\endlist}
\renewcommand\newblock{\hskip .11em\@plus.33em\@minus.07em}
\makeatother

\newenvironment{bottompar}{\par\vspace*{\fill}}{\clearpage}
\renewcommand{\refname}{References}
\DeclareMathOperator*{\argmin}{argmin}
\allowdisplaybreaks

\begin{document}

\title{\textbf{Efficient determination of optimised multi-arm multi-stage experimental designs with control of generalised error-rates}}
\author{\textbf{M. J. Grayling\textsuperscript{1}, J. M. S. Wason\textsuperscript{1,2}, A. P. Mander\textsuperscript{1}}\\
\small 1. Hub for Trials Methodology Research, MRC Biostatistics Unit, Cambridge, UK,
\\ \small 2. Institute of Health and Society, Newcastle University, Newcastle, UK}
\date{}
\maketitle

\noindent \textbf{Running Head:} Generalised multi-arm multi-stage experimental designs.\\

\noindent \textbf{Abstract:} Primarily motivated by the drug development process, several publications have now presented methodology for the design of multi-arm multi-stage experiments with normally distributed outcome variables of known variance. Here, we extend these past considerations to allow the design of what we refer to as an $abcd$ multi-arm multi-stage experiment. We provide a proof of how strong control of the $a$-generalised type-I familywise error-rate can be ensured. We then describe how to attain the power to reject at least $b$ out of $c$ false hypotheses, which is related to controlling the $b$-generalised type-II familywise error-rate. Following this, we detail how a design can be optimised for a scenario in which rejection of any $d$ null hypotheses brings about termination of the experiment. We achieve this by proposing a highly computationally efficient approach for evaluating the performance of a candidate design. Finally, using a real clinical trial as a motivating example, we explore the effect of the design's control parameters on the statistical operating characteristics.\\


\begin{bottompar}
	\noindent Address correspondence to M. J. Grayling, MRC Biostatistics Unit, Forvie Site, Robinson Way, Cambridge CB2 0SR, UK; Fax: +44-(0)1223-330365; E-mail: mjg211@cam.ac.uk. 
\end{bottompar}

\section{Introduction}
\label{s:intro}

The statistical literature contains much research on the design and analysis of two-arm group sequential experiments, with the primary motivation often the desire to improve the efficiency of clinical research (see, for example, Jennison and Turnbull (2000)). Recently, several papers have sought to extend these methods to multi-arm multi-stage (MAMS) designs, which allow multiple arms to be compared to a single control arm. Assuming outcome data to be normally distributed with known variance, Magirr et al. (2012) established how the type-I familywise error-rate (FWER) could be strongly controlled in a parallel arm MAMS design. Allowing early stopping to both reject and accept null hypotheses, with any number of observations accrued in each arm during each stage, they also described how commonly utilised stopping boundaries for two-arm designs could be extended to the multi-arm setting.

Additional recommendations on how to design MAMS experiments have been presented (Wason et al., 2016), along with considerations on other outcome data types and the incorporation of covariates (Jaki and Magirr, 2013). An approach based on simulation was also advanced for the optimisation of a MAMS design (Wason and Jaki, 2012), whilst Ghosh et al. (2017) recently described an effective method for overcoming issues of computational intractability inherent in the approach of Magirr et al. (2012).

Thus, much methodology is now available for the design of MAMS experiments. However, each of these presentations desired to control the conventional type-I FWER, and power to reject a particular null hypothesis under the so-called least favourable configuration (LFC). These operating characteristics are a logical starting point, given their typical requirement in late phase clinical trials. However, allowing for more flexible error control would be highly advantageous. Specifically, there has been increased interest in recent years in experimental designs that control the generalised type-I (Lehmann and Romano, 2005; Romano and Shaikh, 2006) and type-II FWERs (Delorme et al., 2016). These constraints are useful in scenarios where classical error control leads to infeasible or undesirable required sample sizes. 

Moreover, each of the above papers assumed a so-called simultaneous stopping rule in which an experiment would be terminated as soon as a single null hypothesis was rejected. Urach and Posch (2016) recently provided a detailed assessment of the relative merits of simultaneous stopping in comparison to a separate stopping rule. In the latter approach, an experiment is continued until a decision is made for every null hypothesis. These two rules are extreme ones, particularly when the number of arms is large. In many instances we may wish to continue our investigations until some intermediate number of null hypotheses have been rejected. To date though, no methodology has been presented to facilitate this.

Within the context of clinical research, generalised error-rate control and more flexible trial stopping rules may be of most relevance in phase II trials. Explicitly, randomised designs are increasing in use (Ivanova et al., 2016) and popularity in phase II (Sharma et al., 2011; Jung, 2013), but their employment has been hindered because of associated sample size requirements (Pond and Abassi, 2011). Controlling generalised error-rates could assist in overcoming this problem. Furthermore, if many treatments are compared in a MAMS phase II trial, we may not wish to identify the single seemingly most efficacious treatment, but some larger number to carry forward for further exploration.

Therefore, here, we provide methodology for the design of MAMS experiments with control of generalised error criteria under highly adaptable stopping rules. We describe how strong control of the $a$-generalised type-I FWER can be achieved. Next, we describe how power to reject at least $b$ out of $c$ false hypotheses can be attained. Following this, we detail how a design can be optimised for a scenario in which rejection of any $d$ null hypotheses will bring about termination of the experiment. To facilitate obtaining such designs in practice, we also detail an efficient algorithm for evaluating the key operating characteristics of any candidate design. We conclude with an example based on a real trial, and with a short discussion.

\section{Methods} \label{s:model}

\subsection{Hypotheses, analysis, and stopping boundaries}

Although our methodology is relevant to many types of experiment, from here we pose our problem within the context of designing a clinical trial. Thus, our goal will be to test whether several experimental treatments are superior to some common control regimen.

First, define the sets $\mathbb{N}_p=\{q\in\mathbb{N} : \ q \le p\}$ and $\mathbb{N}_p^+=\mathbb{N}_p\backslash\{0\}$, for any $p\in\mathbb{N}^+$. Additionally, for any $\boldsymbol{v}=(v_1,\dots,v_p)^\top\in\mathbb{R}^p$, $p\in\mathbb{N}^+$, and $q\in\mathbb{N}$, define $\boldsymbol{c}(\boldsymbol{v},q)=(\boldsymbol{v}^\top,\dots,\boldsymbol{v}^\top)^\top\in\mathbb{R}^{pq}$ to be the vector formed by repeating $\boldsymbol{v}$, $q$ times. For simplicity, $\boldsymbol{0}_p=\boldsymbol{c}(0, p)$. We will also assume the convention that $\{(v_1,\dots,v_p)^\top\}^{1/2}=(v_1^{1/2},\dots,v_p^{1/2})^\top$, and take $Diag(\boldsymbol{v})$ as the $p\times p$ matrix formed by placing the elements of $\boldsymbol{v}$ along the leading diagonal (i.e., $\{Diag(\boldsymbol{v})\}_{qq}=v_q$, $\{Diag(\boldsymbol{v})\}_{qr}=0$ if $q\neq r$, $q,r\in\mathbb{N}_p^+$). Finally, $\mathbb{I}(A)$ will signify the indicator function on event $A$.

Now, we suppose the trial has $K \in\mathbb{N}^+$ experimental treatments present initially, and let $\mu_k$, $k\in\mathbb{N}_K^+$, be the mean response on the experimental treatments, and $\mu_0$ be the mean response on the control. Our hypotheses of interest are then
$H_{k} : \tau_k = \mu_k - \mu_0 \le 0$, $k\in\mathbb{N}_K^+$. As data are accrued over the course of the trial, this family of null hypotheses will be tested at a series of analyses indexed by $j\in\mathbb{N}_J^+$, for some $J \in\mathbb{N}^+$. We suppose that at any interim analysis $j\in\mathbb{N}_J^+$, to test $H_{k}$, responses $X_{i,l}\sim N(\mu_l,\sigma_l^2)$ will be available from patients $i\in\mathbb{N}_{r_{l,j}n}^+$ on treatments $l\in\{0,k\}$. We enforce that $r_{0,1}=1$, such that $n$ is the group size in the first stage for the control arm, and the $r_{k,j}$ are allocation ratios relative to $n$. We denote the set of all $r_{k,j}n$, given $n$, by $\mathscr{R}_n=\{r_{k,j}n\in\mathbb{R}^+,\ k\in\mathbb{N}_K^+,\ j\in\mathbb{N}_J^+\}$. Note that by this definition of $\mathscr{R}_n$, for now we require only that $r_{k,j}n\in\mathbb{R}^+$. We return later to discuss how we can ensure $r_{k,j}n\in\mathbb{N}^+$.

The test statistics $Z_{k,j|\mathscr{R}_n}\equiv Z_{k,j}=\hat{\tau}_{k,j|\mathscr{R}_n}\{Var(\hat{\tau}_{k,j|\mathscr{R}_n})\}^{-1/2}=\hat{\tau}_{k,j|\mathscr{R}_n}I_{k,j|\mathscr{R}_n}^{1/2}$ ($k\in\mathbb{N}_K^+$, $j\in\mathbb{N}_J^+$) will be used, where $\hat{\tau}_{k,j|\mathscr{R}_n}\equiv \hat{\tau}_{k,j} = (r_{k,j}n)^{-1}\sum_{i=1}^{r_{k,j}n}X_{i,k} - (r_{0,j}n)^{-1}\sum_{i=1}^{r_{0,j}n}X_{0,k}$. Here, we have initially explicitly stated the dependence upon $\mathscr{R}_n$, but will routinely omit it in what follows for brevity.

Similarly, we set $\boldsymbol{\tau}=(\tau_1,\dots,\tau_K)^\top$ and $\hat{\boldsymbol{\tau}}_{j|\mathscr{R}_n}\equiv \hat{\boldsymbol{\tau}}_{j}=(\hat{\tau}_{1,j},\dots,\hat{\tau}_{K,j})^\top$ for $j\in\mathbb{N}_J^+$. Moreover, take $\boldsymbol{Z}_{|\mathscr{R}_n}\equiv \boldsymbol{Z}=(\boldsymbol{Z}_{1},\dots,\boldsymbol{Z}_{J})^\top$ and $\boldsymbol{I}_{|\mathscr{R}_n}\equiv\boldsymbol{I}=(\boldsymbol{I}_{1}^\top,\dots,\boldsymbol{I}_{J}^\top)^\top$, for $\boldsymbol{Z}_{j|\mathscr{R}_n}\equiv\boldsymbol{Z}_{j}=(Z_{1,j},\dots,Z_{K,j})^\top$ and $\boldsymbol{I}_{j|\mathscr{R}_n}\equiv\boldsymbol{I}_{j}=(I_{1,j|\mathscr{R}_n},\dots,I_{K,j|\mathscr{R}_n})^\top$. Then, $\boldsymbol{Z}$ has a multivariate normal distribution (Jennison and Turnbull, 2000) with
\begin{small}
	\begin{align*}
	\mathbb{E}(\boldsymbol{Z}) &= \boldsymbol{c}(\boldsymbol{\tau}, J)\circ\boldsymbol{I}^{1/2},\\
	Cov(\boldsymbol{Z},\boldsymbol{Z}) &= \begin{Bmatrix} Cov(\boldsymbol{Z}_{1},\boldsymbol{Z}_{1}) & \dots & Cov(\boldsymbol{Z}_{1},\boldsymbol{Z}_{J}) \\ \vdots & \ddots & \vdots \\ Cov(\boldsymbol{Z}_{J},\boldsymbol{Z}_{1}) & \dots & Cov(\boldsymbol{Z}_{J},\boldsymbol{Z}_{J})\end{Bmatrix},\\
	Cov(\boldsymbol{Z}_{j_1},\boldsymbol{Z}_{j_2}) &= Cov(\hat{\boldsymbol{\tau}}_{j_1}\circ\boldsymbol{I}_{j_1}^{1/2},\hat{\boldsymbol{\tau}}_{j_2}\circ\boldsymbol{I}_{j_2}^{1/2})=Diag(\boldsymbol{I}_{j_1}^{1/2})Cov(\hat{\boldsymbol{\tau}}_{j_2},\hat{\boldsymbol{\tau}}_{j_2})Diag(\boldsymbol{I}_{j_2}^{1/2}), \qquad j_2\in\mathbb{N}_J^+,\ j_1\in\mathbb{N}_{j_2}^+,\\
	Cov(\hat{\boldsymbol{\tau}}_{j},\hat{\boldsymbol{\tau}}_{j}) &= \begin{Bmatrix} \sigma_0^2/(r_{0,j}n) + \sigma_1^2/(r_{1,j}n) & \sigma_0^2/(r_{0,j}n) & \dots & \sigma_0^2/(r_{0,j}n) \\ \sigma_0^2/(r_{0,j}n)& \ddots &\ddots & \vdots \\ \vdots & \ddots & \ddots & \sigma_0^2/(r_{0,j}n)\\
	\sigma_0^2/(r_{0,j}n) & \dots & \sigma_0^2/(r_{0,j}n) & \sigma_0^2/(r_{0,j}n) + \sigma_K^2/(r_{K,j}n) \end{Bmatrix}, \qquad j\in\mathbb{N}_{J}^+.
	\end{align*}
\end{small}

As discussed, we suppose that we would like our design to strongly control the $a$-generalised type-I FWER, the probability of incorrectly rejecting at least $a\in\mathbb{N}_K^+$ true null hypotheses, to some level $\alpha\in(0,1)$. Furthermore, we extend the typically employed LFC considerations of Dunnett (1984) to the requirement of Dunnett and Tamhane (1992), and assume it is desired that the $b$ out of $c$ familywise power (FWP) be at least $1-\beta\in(0,1)$, for some $c\in\mathbb{N}_K^+$ and $b\in\mathbb{N}_c^+$. That is, without loss of generality, the probability of rejecting at least $b$ of the null hypotheses $H_{1},\dots,H_{c}$, when $\tau_k=\delta\in\mathbb{R}^+$ and $\tau_{k'}=\delta_0<\delta$ for $k\in\mathbb{N}_c^+$ and $k'\in\mathbb{N}_K\backslash\mathbb{N}_c$, must be at least $1-\beta$. Here, $\delta$ is interpreted as an effect size at which it would be of interest to consider a treatment further, whilst if the treatment effect is below $\delta_0$ then this arm is not worth considering further. Note that this FWP requirement is directly related to the generalised type-II FWER defined in Delorme et al. (2016), who use the notation $r$ and $p$, where we use $b$ and $c$. In what follows, we set $\boldsymbol{\delta}_{c,K}=\{\boldsymbol{c}(\delta,c)^\top,\boldsymbol{c}(\delta_0,K-c)^\top\}^\top$.

Design determination is also achieved supposing that the cumulative rejection of $d\in\mathbb{N}_K^+$ or more null hypotheses, by any analysis $j\in\mathbb{N}_J^+$, will bring about the termination of the trial. With this, $d=1$ and $d=K$ correspond to the simultaneous and separate stopping rules respectively, whilst $d\in\mathbb{N}_{K}^+\backslash\{1,K\}$ represents a new previously unconsidered stopping rule.

We refer to our designs from here as $abcd$-MAMS designs. The design considered in Magirr et al. (2012), for example, is then a 1111-MAMS design. Explicitly, in a 1111-MAMS design, the conventional FWER is strongly controlled, the trial is terminated as soon as any null hypothesis is rejected, and power is provided to reject a particular treatment's null hypothesis when only it is worthy of further consideration. Jung (2008) explored, in the context of exact binomial tests, a 11$K$1-MAMS design. With this, power is instead provided to reject at least one null hypothesis amongst the entire family, when all treatments have an effectiveness considered interesting.

We define our futility (non-rejection) and efficacy (rejection) stopping boundaries as $\boldsymbol{f}=(f_1,\dots,f_J)^\top$ and $\boldsymbol{e}=(e_1,\dots,e_J)^\top$ respectively. Allowing for infinite stopping boundaries to preclude the possibility to make early decisions to reject or accept null hypotheses if so desired, we denote the set of all possible pairs $(\boldsymbol{f},\boldsymbol{e})$ by $\mathscr{B}=\left[ (\boldsymbol{f},\boldsymbol{e})\in\{\mathbb{R}\cup(-\infty)\}^J\times\{\mathbb{R}\cup(\infty)\}^J : \forall j\in\mathbb{N}_{J-1}^+\ f_j< e_j,\ f_J=e_J\in\mathbb{R}\right]$. Here, $f_J=e_J$ is enforced so that the trial terminates after at most $J$ stages as desired.

Finally, we define our trials formal conduct using two vectors, $\boldsymbol{\psi}=(\psi_1,\dots,\psi_K)^\top$ and $\boldsymbol{\omega}=(\omega_1,\dots,\omega_K)^\top$, as follows
\begin{enumerate}
\item Set $\boldsymbol{\psi}=\boldsymbol{\omega}=\boldsymbol{0}_K$ and $j=1$.
\item Conduct stage $j$ of the trial, and compute the $Z_{kj}$.
\item For $k\in\mathbb{N}_K^+$ such that $Z_{k,l}\in (f_l,e_l]$ for $l\in\mathbb{N}_{j-1}$ (with the convention $Z_{k,0}\in (f_0,e_0]$ for $k\in\mathbb{N}_K^+$)
\begin{itemize}
\item If $Z_{k,j}>e_j $ reject $H_{k}$, setting $\psi_{k}=1$ and $\omega_{k}=j$, and designating that $Z_{k,l}=\infty$ for $l\in\mathbb{N}_J\backslash\mathbb{N}_j$.
\item If $Z_{k,j}\le f_j$ do not reject $H_{k}$, setting $\omega_{k}=j$, and designating that $Z_{k,l}=-\infty$ for $l\in\mathbb{N}_J\backslash\mathbb{N}_j$.
\end{itemize}
\item If $\sum_{k=1}^{K}\mathbb{I}(\psi_{k}=1) < d$ and $\sum_{k=1}^{K}\mathbb{I}(\omega_{k}=0) > 0$, set $j=j+1$ and return to 2. Else stop the trial, and for each $k\in\mathbb{N}_K^+$ with $\omega_{k}=0$, set $\omega_{k}=j$, and designate that $Z_{k,l}=-\infty$ for $l\in\mathbb{N}_J\backslash\mathbb{N}_j$.
\end{enumerate}
Here, if fewer than $d$ null hypotheses have been rejected, and a decision has not been made for all $K$ null hypotheses, step 4 leads to a return to step 2.

\subsection{Strong control of the $a$-generalised type-I familywise error-rate}

To demonstrate how strong control of the $a$-generalised type-I FWER can be achieved, we extend the arguments of Magirr et al. (2012). First, for $\boldsymbol{\theta}=(\theta_1,\dots,\theta_K)\in\mathbb{R}^K$, $k\in\mathbb{N}_K^+$, and $j\in\mathbb{N}_J^+$, define
\begin{align*}
F_{k,j}(\theta_k) &= \{ Z_{k,j} \le f_j + (\tau_k - \theta_k)I_{k,j}^{1/2}\},\\
C_{k,j}(\theta_k) &= \{ f_j + (\tau_k - \theta_k)I_{k,j}^{1/2} < Z_{k,j} \le e_j + (\tau_k - \theta_k)I_{k,j}^{1/2}\}.
\end{align*}
Next, take $S_{a,K} = \left\{(s_1,\dots,s_p)\in(\mathbb{N}_K^+)^p,\ p\in\mathbb{N}_{K}\backslash\mathbb{N}_{K-a} : \forall (i,j)\in\mathbb{N}_{p-1}^+\times(\mathbb{N}_{p}\backslash\mathbb{N}_{i})\ s_i<s_j \right\}$. That is, $S_{a,K}$ is the set of all ordered subsets of $\mathbb{N}_K^+$ with $|S_{a,K}|\ge K-a+1$.

Then, if $\tau_k=\theta_k$ for $k\in\mathbb{N}_K^+$, the event that at least $K-a+1$ of $H_{1},\dots,H_{K}$ fail to be rejected is equivalent to
\[ \bar{R}_{a,K}(\boldsymbol{\theta})=\bigcup_{S\in S_{a,K}}\left\{ \bigcap_{k\in S}\left( \bigcup_{j=1}^J \left[ \left\{ \bigcap_{l=1}^{j-1}C_{k,l}(\theta_k) \right\}\cap F_{k,j}(\theta_k) \right] \right) \right\}, \]
taking $\cap_{l=1}^0=U$, where $U$ is the whole sample space.

\begin{lemma}\label{lemma1}
	For any $\epsilon_k>0$
	\[ \bigcup_{j=1}^J \left[ \left\{ \bigcap_{l=1}^{j-1}C_{k,l}(\theta_k+\epsilon_k) \right\}\cap F_{k,j}(\theta_k+\epsilon_k) \right]\subseteq \bigcup_{j=1}^J \left[ \left\{ \bigcap_{l=1}^{j-1}C_{k,l}(\theta_k) \right\}\cap F_{k,j}(\theta_k) \right]. \]
\end{lemma}
\begin{proof}
	See the Appendix of Magirr et al. (2012).
\end{proof}

\begin{thm}\label{thm1}
	For any $\boldsymbol{\theta}$, $\mathbb{P}(\text{Reject at least }a\text{ true }H_{k}|\boldsymbol{\theta})\le\mathbb{P}(\text{Reject at least }a\text{ true }H_{k}|\boldsymbol{0}_K)$.
\end{thm}
\begin{proof}
	Suppose without loss of generality that $\theta_p\le0$ and $\theta_q>0$ for $p\in\mathbb{N}_l^+$, $l\in\mathbb{N}_K$ and $q\in\mathbb{N}_K\backslash\mathbb{N}_{l}$. Let $\boldsymbol{\theta}_l=(\theta_1,\dots,\theta_l)^\top$. Using Lemma~\ref{lemma1}
	\begin{align*}
	\mathbb{P}(\text{Reject at least }a\text{ true }H_{k}|\boldsymbol{\theta})&\le \mathbb{P}(Z_{k,j}>e_j\text{ for some }j\in\mathbb{N}_J^+ \text{ for at least } a \ k\in\mathbb{N}_l^+|\boldsymbol{\theta}),\\
	&= 1 - \mathbb{P}\{\bar{R}_{a,l}(\boldsymbol{0}_l)\},\\
	&\le 1 - \mathbb{P}\{\bar{R}_{a,l}(\boldsymbol{0}_K)\},\\
	&\le 1 - \mathbb{P}\{\bar{R}_{a,K}(\boldsymbol{0}_K)\},\\
	&=\mathbb{P}(Z_{k,j}>e_j\text{ for some }j\in\mathbb{N}_J^+ \text{ for at least } a \ k\in\mathbb{N}_K^+|\boldsymbol{0}_K).\\
	&=\mathbb{P}(\text{Reject at least }a\text{ true }H_{k}|\boldsymbol{0}_K).
	\end{align*}
\end{proof}

Thus, by Theorem~\ref{thm1}, the $a$-generalised type-I FWER can be controlled in the strong sense to level $\alpha$ by ensuring $1-\mathbb{P}\{\bar{R}_{a,K}(\boldsymbol{0}_K)\}\le\alpha$.

\subsection{Design operating characteristics}\label{desdet}

In this section, we describe an efficient approach for the evaluation of the key performance characteristics of any design.

Ghosh et al. (2017) discussed how the approach utilised in Magirr et al. (2012) is prohibitively computationally expensive even for moderate values of the number of stages or treatments. By working with the score statistics, and employing recently developed techniques for the efficient evaluation of multivariate normal integrals, they provided a method which was considerably more efficient. However, their approach makes no allowance for the fast evaluation of expected sample sizes, as it is geared towards assessing the probability one of the null hypotheses is rejected, and does not monitor the outcome for each arm. It is also specialised to the case $d=1$. Though it could in theory be extended to $d>1$ designs, to allow the optimisation of the stopping boundaries and group sizes in Section~\ref{optimal}, we here propose a different approach based on the distribution of $\boldsymbol{Z}$, the standardised Wald test statistics.

We denote the possible outcomes of the trial by $\{\boldsymbol{\Psi}=(\Psi_1,\dots,\Psi_K)^\top,\boldsymbol{\Omega}=(\Omega_1,\dots,\Omega_K)^\top\}$, where
\begin{itemize}
	\item \(\Psi_k \in \mathbb{N}_1\), with \(\Psi_{k}=1\) if \(H_{k}\) is rejected, and \(\Psi_{k}=0\) otherwise,
	\item \(\Omega_{k} \in \mathbb{N}_J^+\), with \(\Omega_{k}=j\) if $j$ is the interim analysis at which $H_{k}$ is rejected or accepted, or the whole trial is stopped and no decision is made on $H_{k}$.
\end{itemize}
With this definition, $(\boldsymbol{\psi},\boldsymbol{\omega})$ from the presented trial conduct is then the realised value of $(\boldsymbol{\Psi},\boldsymbol{\Omega})$. We denote the sample space for $(\boldsymbol{\Psi},\boldsymbol{\Omega})$, for any $(\boldsymbol{f},\boldsymbol{e})\in\mathscr{B}$, given $d$, $J$ and $K$, by $\Xi(\boldsymbol{f},\boldsymbol{e}|d,J,K)$. We present the exact definition of $\Xi(\boldsymbol{f},\boldsymbol{e}|d,J,K)$ in the Appendix.

The probability of a particular outcome $(\boldsymbol{\Psi},\boldsymbol{\Omega})=(\boldsymbol{\psi},\boldsymbol{\omega})\in\Xi(\boldsymbol{f},\boldsymbol{e}|d,J,K)$ is then
\begin{normalsize}
	\begin{align*}
	\mathbb{P}(\boldsymbol{\Psi}=\boldsymbol{\psi},\boldsymbol{\Omega}=\boldsymbol{\omega}\mid\boldsymbol{\tau},\boldsymbol{f},\boldsymbol{e},J,K,\mathscr{R}_n) &= \int_{\text{l}(1,1,\boldsymbol{\psi},\boldsymbol{\omega},\boldsymbol{f},\boldsymbol{e})}^{\text{u}(1,1,\boldsymbol{\psi},\boldsymbol{\omega},\boldsymbol{f},\boldsymbol{e})} \dots \int_{\text{l}(K,1,\boldsymbol{\psi},\boldsymbol{\omega},\boldsymbol{f},\boldsymbol{e})}^{\text{u}(K,1,\boldsymbol{\psi},\boldsymbol{\omega},\boldsymbol{f},\boldsymbol{e})} \dots \int_{\text{l}(1,J,\boldsymbol{\psi},\boldsymbol{\omega},\boldsymbol{f},\boldsymbol{e})}^{\text{u}(1,J,\boldsymbol{\psi},\boldsymbol{\omega},\boldsymbol{f},\boldsymbol{e})} \dots  \\ & \qquad \dots \int_{\text{l}(K,J,\boldsymbol{\psi},\boldsymbol{\omega},\boldsymbol{f},\boldsymbol{e})}^{\text{u}(K,J,\boldsymbol{\psi},\boldsymbol{\omega},\boldsymbol{f},\boldsymbol{e})} \! \phi\left\{\boldsymbol{x},\boldsymbol{c}(\boldsymbol{\tau},J) \circ \boldsymbol{I}^{1/2}, Cov(\boldsymbol{Z},\boldsymbol{Z})\right\} \ \\
	& \qquad \qquad \qquad \mathrm{d}x_{KJ}\dots \mathrm{d}x_{1J}\dots\mathrm{d}x_{K1}\dots\mathrm{d}x_{11},\label{probsni}
	\end{align*}
\end{normalsize}
where $\phi\{\boldsymbol{x},\boldsymbol{c}(\boldsymbol{\tau},J) \circ \boldsymbol{I}^{1/2}, Cov(\boldsymbol{Z},\boldsymbol{Z})\}$ is the probability density function of a multivariate normal distribution with mean $\boldsymbol{c}(\boldsymbol{\tau},J) \circ \boldsymbol{I}^{1/2}$ and covariance matrix $Cov(\boldsymbol{Z},\boldsymbol{Z})$, evaluated at vector $\boldsymbol{x}=(x_{11},\dots,x_{KJ})^\top$. Furthermore
\begin{align*}
\text{l}(k,j,\boldsymbol{\psi},\boldsymbol{\omega},\boldsymbol{f},\boldsymbol{e}) &= \begin{cases} e_j &: \psi_k=1,\ \omega_k=j, \\ -\infty &: \{\omega_k<j\}\cup\{\psi_k=0, \omega_k=j\}, \\ f_j \hphantom{a+}  &: \text{otherwise}, \end{cases}\\
\text{u}(k,j,\boldsymbol{\psi},\boldsymbol{\omega},\boldsymbol{f},\boldsymbol{e}) &= \begin{cases} e_j &: \{\omega_k>j\}\cup\{\psi_k=0, \omega_k=j,\max_{m\in\mathbb{N}_K^+}\omega_m=j,\ \Sigma_{m\in\mathbb{N}_K^+}\mathbb{I}(\psi_m=1)\ge d\}, \\ \infty &: \{\omega_k<j\}\cup\{\psi_k=1,\ \omega_k=j\}, \\ f_j \hphantom{a+} &: \text{otherwise}, \end{cases}
\end{align*}

for $k\in\mathbb{N}_K^+$ and $j\in\mathbb{N}_J^+$. Note that the above integral is presented for notational simplicity as being $JK$ dimensional. However, using the marginal distribution properties of the multivariate normal distribution, it can be reduced immediately by removing integrals, and the corresponding elements of $\phi(\cdot)$, for which the range of integration is $(-\infty,\infty)$. These correspond exactly to those $Z_{k,j}$ designated to be $\pm\infty$ by the trial's formal conduct above.

Denoting the trial's (random) required sample size by $N$, the expected sample size (ESS) can then be computed as
\begin{equation}
\mathbb{E}(N\mid\boldsymbol{\tau},\boldsymbol{f},\boldsymbol{e},d,J,K,\mathscr{R}_n) = \sum_{(\boldsymbol{\psi},\boldsymbol{\omega}) \in \Xi(\boldsymbol{f},\boldsymbol{e}|d,J,K)} n\left( \max_{k\in\mathbb{N}_K^+}\omega_k + \sum_{k=1}^Kr_{k,\omega_k}\right)\mathbb{P}(\boldsymbol{\psi},\boldsymbol{\omega}\mid\boldsymbol{\tau},\boldsymbol{f},\boldsymbol{e},J,K,\mathscr{R}_n).\label{ess}
\end{equation}
Whilst the probabilities of committing a $p$-generalised type-I familywise error, and of rejecting at least $p$ of $H_1,\dots,H_q$, respectively, are
\begin{align}
FWER_I(p\mid \boldsymbol{\tau},\boldsymbol{f},\boldsymbol{e},d,J,K,\mathscr{R}_n)&= \sum_{(\boldsymbol{\psi},\boldsymbol{\omega}) \in \Xi_{\text{FWER}_I}(p,\boldsymbol{\tau}\mid \boldsymbol{f},\boldsymbol{e},d,J,K)} \mathbb{P}(\boldsymbol{\psi},\boldsymbol{\omega}\mid\boldsymbol{\tau},\boldsymbol{f},\boldsymbol{e},J,K,\mathscr{R}_n),\label{fwer}\\
FWP(p,q\mid \boldsymbol{\tau},\boldsymbol{f},\boldsymbol{e},d,J,K,\mathscr{R}_n)&= \sum_{(\boldsymbol{\psi},\boldsymbol{\omega}) \in \Xi_{\text{FWP}}(p,q\mid \boldsymbol{f},\boldsymbol{e},d,J,K)} \mathbb{P}(\boldsymbol{\psi},\boldsymbol{\omega}\mid\boldsymbol{\tau},\boldsymbol{f},\boldsymbol{e},J,K,\mathscr{R}_n),\label{power}
\end{align}
for
\begin{align*}
\Xi_{\text{FWER}_I}(p,\boldsymbol{\tau}\mid \boldsymbol{f},\boldsymbol{e},d,J,K) &= \left\{(\boldsymbol{\psi},\boldsymbol{\omega}) \in \Xi(\boldsymbol{f},\boldsymbol{e}|d,J,K) : \sum_{k=1}^{K}\mathbb{I}(\psi_k = 1)\mathbb{I}(\tau_k \le 0) \ge p\right\},\\
\Xi_{\text{FWP}}(p,q\mid \boldsymbol{f},\boldsymbol{e},d,J,K) &= \left\{(\boldsymbol{\psi},\boldsymbol{\omega}) \in \Xi(\boldsymbol{f},\boldsymbol{e}|d,J,K) : \sum_{k=1}^{q}\mathbb{I}(\psi_k = 1) \ge p\right\}.
\end{align*}
Equations~(\ref{ess})-(\ref{power}) together provide the operating characteristics of candidate designs that are typically required for the determination of stopping boundaries and group sizes. Therefore, we have technically described all that is required for the determination of $abcd$-MAMS designs. However, particularly when the number of stages or treatment arms is large, it is important to be as efficient as possible in evaluating these characteristics. For this reason, we now describe how one may more improve upon the above in certain routinely faced design scenarios. 

Specifically, suppose for two treatment arms $k_1$ and $k_2$, $k_1,k_2\in\mathbb{N}_K^+$, $k_1\neq k_2$, that $\tau_{k_1}=\tau_{k_2}$, $r_{k_1,j}=r_{k_2,j}$ for $j\in\mathbb{N}_J^+$, and $\sigma_{k_1}^2=\sigma_{k_2}^2$. Then these arms are interchangeable in the sense that
\begin{align*}
\mathbb{P}(\boldsymbol{\psi},\boldsymbol{\omega}|\boldsymbol{\tau},\boldsymbol{f},\boldsymbol{e},J,K,\mathscr{R}_n)=\mathbb{P}(\boldsymbol{\psi}',\boldsymbol{\omega}'|\boldsymbol{\tau},\boldsymbol{f},\boldsymbol{e},J,K,\mathscr{R}_n),
\end{align*}
for $(\boldsymbol{\psi},\boldsymbol{\omega}),(\boldsymbol{\psi}',\boldsymbol{\omega}')\in\Xi(\boldsymbol{f},\boldsymbol{e}|d,J,K)$ if $\psi_{k_1}'=\psi_{k_2}$, $\psi_{k_2}'=\psi_{k_1}$, $\omega_{k_1}'=\omega_{k_2}$, $\omega_{k_2}'=\omega_{k_1}$, and $\psi_k'=\psi_k$, $\omega_k'=\omega_k$ for $\mathbb{N}_K^+\backslash\{k_1,k_2\}$. This allows us to reduce the number of integrals that must be considered. Here, we explain how this is achieved when $\sigma_1^2=\dots=\sigma_K^2$ and $r_{1,j}=\dots=r_{K,j}$ for $j\in\mathbb{N}_J^+$. Methodology for other scenarios should then be clear.

We place an order upon the outcome for the experimental treatments arms, setting
\begin{equation*}\label{xiEF2}
\begin{split}
& \Xi'(\boldsymbol{f},\boldsymbol{e},\boldsymbol{\tau}|d,J,K) =\left\{(\boldsymbol{\psi},\boldsymbol{\omega})\in\Xi(\boldsymbol{f},\boldsymbol{e}|d,J,K) : \forall (k_1,k_2)\in(\mathbb{N}_K^+)^2 \ \tau_{k_1}=\tau_{k_2}\right.\\
& \left.\qquad\qquad\qquad\qquad\qquad\qquad\qquad\qquad\qquad \Longrightarrow2\omega_{k_1}-\mathbb{I}(\psi_{k_1}=0)\ge 2\omega_{k_2}-\mathbb{I}(\psi_{k_2}=0) \right\}.
\end{split}
\end{equation*}
In the Appendix, we expand on the intuition behind this definition. With it however, we have
\begin{footnotesize}
	\begin{align}
	FWER_I(p\mid\boldsymbol{\tau}, \boldsymbol{f},\boldsymbol{e},d,J,K,\mathscr{R}_n) &= \sum_{(\boldsymbol{\psi},\boldsymbol{\omega}) \in \Xi_{\text{FWER}_I}'(p,\boldsymbol{\tau}\mid \boldsymbol{f},\boldsymbol{e},d,J,K)} \text{d}(\boldsymbol{\psi},\boldsymbol{\omega}|\boldsymbol{\tau},\boldsymbol{f},\boldsymbol{e},d,J,K)\mathbb{P}(\boldsymbol{\psi},\boldsymbol{\omega}\mid\boldsymbol{\tau},\boldsymbol{f},\boldsymbol{e},J,K,\mathscr{R}_n),\label{fwer2}\\
	FWP(p,q\mid \boldsymbol{\tau},\boldsymbol{f},\boldsymbol{e},d,J,K,\mathscr{R}_n) &= \sum_{(\boldsymbol{\psi},\boldsymbol{\omega}) \in \Xi_{\text{FWP}}'(p,q,\boldsymbol{\tau}\mid \boldsymbol{f},\boldsymbol{e},d,J,K)} \text{d}(\boldsymbol{\psi},\boldsymbol{\omega}|\boldsymbol{\tau},\boldsymbol{f},\boldsymbol{e},d,J,K)\mathbb{P}(\boldsymbol{\psi},\boldsymbol{\omega}\mid\boldsymbol{\tau},\boldsymbol{f},\boldsymbol{e},J,K,\mathscr{R}_n)\label{power2},\\
	\begin{split}
	\mathbb{E}(N\mid\boldsymbol{\tau},\boldsymbol{f},\boldsymbol{e},d,J,K,\mathscr{R}_n) &= \sum_{(\boldsymbol{\psi},\boldsymbol{\omega}) \in \Xi'(\boldsymbol{f},\boldsymbol{e},\boldsymbol{\tau}|d,J,K)} n\left( \max_{k\in\mathbb{N}_K^+}\omega_k + \sum_{k=1}^Kr_{k,\omega_k}\right)\text{d}(\boldsymbol{\psi},\boldsymbol{\omega}|\boldsymbol{\tau},\boldsymbol{f},\boldsymbol{e},d,J,K)\\
	&\qquad\qquad\qquad\qquad\qquad\qquad\times\mathbb{P}(\boldsymbol{\psi},\boldsymbol{\omega}\mid\boldsymbol{\tau},\boldsymbol{f},\boldsymbol{e},J,K,\mathscr{R}_n),\label{ess2}
	\end{split}
	\end{align}
\end{footnotesize}
for
\begin{align*}
\Xi_{\text{FWER}_I}'(p,\boldsymbol{\tau}\mid \boldsymbol{f},\boldsymbol{e},d,J,K) &= \left\{(\boldsymbol{\psi},\boldsymbol{\omega}) \in \Xi'(\boldsymbol{f},\boldsymbol{e},\boldsymbol{\tau}|d,J,K) : \sum_{k=1}^{K}\mathbb{I}(\psi_k = 1)\mathbb{I}(\tau_k \le 0) \ge p\right\},\\
\Xi_{\text{FWP}}'(p,q,\boldsymbol{\tau}\mid \boldsymbol{f},\boldsymbol{e},d,J,K) &= \left\{(\boldsymbol{\psi},\boldsymbol{\omega}) \in \Xi'(\boldsymbol{f},\boldsymbol{e},\boldsymbol{\tau}|d,J,K) : \sum_{k=1}^{q}\mathbb{I}(\psi_k = 1) \ge p\right\},
\end{align*}
where $\text{d}(\boldsymbol{\psi},\boldsymbol{\omega}|\boldsymbol{\tau},\boldsymbol{f},\boldsymbol{e},d,J,K)$ is the degeneracy of scenario $(\boldsymbol{\psi},\boldsymbol{\omega})$ given $\boldsymbol{\tau}$, $\boldsymbol{f}$, $\boldsymbol{e}$, $d$, $J$, and $K$. That is, $\text{d}(\boldsymbol{\psi},\boldsymbol{\omega}|\boldsymbol{\tau},\boldsymbol{f},\boldsymbol{e},d,J,K)$ is the number of possible scenarios in $\Xi(\boldsymbol{f},\boldsymbol{e}|d,J,K)$, which could be produced from $(\boldsymbol{\psi},\boldsymbol{\omega})$, that would have equal probability because of the equality of treatment effects in the vector $\boldsymbol{\tau}$.

To evaluate the efficiency improvements made by accounting for any interchangeability amongst the experimental treatments, one can compare the relative sizes of the sets $\Xi(\boldsymbol{f},\boldsymbol{e}|d,J,K)$ and $\Xi'(\boldsymbol{f},\boldsymbol{e},\boldsymbol{\tau}|d,J,K)$. Unfortunately, the values of $|\Xi(\boldsymbol{f},\boldsymbol{e}|d,J,K)|$ and $|\Xi'(\boldsymbol{f},\boldsymbol{e},\boldsymbol{\tau}|d,J,K)|$ do not in general have a simple form. However, in the Appendix we present formulae for their values, along with the $\text{d}(\boldsymbol{\psi},\boldsymbol{\omega}|\boldsymbol{\tau},\boldsymbol{f},\boldsymbol{e},d,J,K)$, when $\boldsymbol{\tau}=\boldsymbol{0}_K$ and when $\boldsymbol{\tau}=\boldsymbol{\delta}_{c,K}$, for the scenario in which early stopping to accept or reject null hypotheses is possible at each analysis (i.e., $f_j,e_j\in\mathbb{R}$ for $j\in\mathbb{N}_J^+$). Therefore as an example, we compare the relative sizes of these sets in this setting when $\boldsymbol{\tau}=\boldsymbol{0}_K$ in Table~\ref{tab1}. It is clear that the value of $|\Xi'(\boldsymbol{f},\boldsymbol{e},\boldsymbol{0}_K|d,J,K)|$ becomes substantially smaller than the corresponding  $|\Xi(\boldsymbol{f},\boldsymbol{e}|d,J,K)|$ as the number of treatments or stages grows. This highlights the importance of utilising this refined approach whenever possible.
	
\def\arraystretch{1.2}
\begin{table}[htb]
	\begin{center}
		\caption{A comparison of the sizes of the sets $\Xi(\boldsymbol{f},\boldsymbol{e}|d,J,K)$ and $\Xi'(\boldsymbol{f},\boldsymbol{e},\boldsymbol{0}_K|d,J,K)$, for several values of $d$, $J$, and $K$.}
		\label{tab1}
		\begin{tabular}{rrrrrrrrr}
			\hline
			& & \multicolumn{3}{c}{$|\Xi(\boldsymbol{f},\boldsymbol{e}|d,J,K)|$} && \multicolumn{3}{c}{$|\Xi'(\boldsymbol{f},\boldsymbol{e},\boldsymbol{0}_K|d,J,K)|$} \\
			\cline{3-5}\cline{7-9}
			$d$ & $K$ & $J=2$ & $J=3$ & $J=4$ && $J=2$ & $J=3$ & $J=4$\\
			\hline
			1 & 2 & 12 & 24 & 40 && 8 & 15 & 24 \\
			2 & 2 & 16 & 36 & 64 && 10 & 21 & 36 \\
			1 & 3 & 34 & 90 & 188 && 13 & 29 & 54 \\
			2 & 3 & 58 & 186 & 428 && 18 & 48 & 100 \\
			3 & 3 & 64 & 216 & 512 && 20 & 56 & 120 \\
			1 & 4 & 96 & 336 & 880 && 19 & 49 & 104\\
			2 & 4 & 200 & 888 & 2608 && 28 & 90 & 220 \\
			3 & 4 & 248 & 1224 & 3808 && 33 & 116 & 300 \\
			4 & 4 & 256 & 1296 & 4096 && 35 & 126 & 330 \\
			\hline
		\end{tabular}
	\end{center}
\end{table}

\subsection{Optimal generalised multi-arm multi-stage designs}\label{optimal}

Stopping boundaries and group sizes providing particular type-I and type-II FWERs can be determined for $abcd$-MAMS trial designs in a variety of ways. An adaptation of the original error spending approach to group sequential trial design (Lan and DeMets, 1983) could be employed, or a functional form could be assumed for the boundaries. In the Appendix, we describe how the triangular test of Whitehead and Stratton (1983) can be modified for $abcd$-MAMS designs. Here, our focus will be on determining optimised stopping boundaries, in order to attempt to maximise trial efficiency.

We suppose that $\sigma_k^2$ and $r_{k,j}$ have been specified for $j\in\mathbb{N}_J^+$ and $k\in\mathbb{N}_K$. Furthermore, we assume we desire to allow early stopping for either efficacy or futility at each analysis, and that values for $a$, $b$, $c$, $d$, $J$, $K$, $\delta$, $\delta_0$, $\alpha$, and $\beta$, have all been designated. Adapting Wason et al. (2012), our goal will be, for $\mathscr{B}'= \left\{ (\boldsymbol{f},\boldsymbol{e})\in\mathbb{R}^J\times\mathbb{R}^J : \forall j\in\mathbb{N}_{J-1}^+\ f_j< e_j,\ f_J=e_J\in\mathbb{R}\right\}$, to find the solution to the following optimisation problem
\begin{equation*}
\argmin_{\{n, (\boldsymbol{f}, \boldsymbol{e})\} \in \mathbb{R}^+\times\mathscr{B}'} w_1\mathbb{E}(N\mid\boldsymbol{0}_K,\boldsymbol{f},\boldsymbol{e},d,J,K,\mathscr{R}_n) + w_2\mathbb{E}(N\mid\boldsymbol{\delta}_{c,K},\boldsymbol{f},\boldsymbol{e},d,J,K,\mathscr{R}_n) + w_3n\sum_{k=0}^Kr_{k,J},
\end{equation*}\label{optcriteria}
subject to
\begin{align*}
FWER_I(a\mid\boldsymbol{0}_K, \boldsymbol{f},\boldsymbol{e},d,J,K,\mathscr{R}_n)&\le \alpha,\\
FWP(b,c\mid \boldsymbol{\delta}_{c,K},\boldsymbol{f},\boldsymbol{e},d,J,K,\mathscr{R}_n)&\ge 1-\beta,
\end{align*}

for $w_1,w_2,w_3\in\mathbb{R}^+\cup\{0\}$. Note that other optimality criteria could be treated similarly. Additionally, in general it is wise to enforce that $w_1+w_2>0$, since the optimal design in the case $w_1=w_2=0$ would be the one converging toward a single-stage design.

The complexity of the above constraints prevents their direct incorporation in to the search procedure. Therefore, we translate our problem to search for the solution to
\begin{small}
	\begin{align*}
	\argmin_{\{n, (\boldsymbol{f}, \boldsymbol{e})\} \in \mathbb{R}^+\times\mathscr{B}'} & w_1\mathbb{E}(N\mid\boldsymbol{0}_K,\boldsymbol{f},\boldsymbol{e},d,J,K,\mathscr{R}_n) + w_2\mathbb{E}(N\mid\boldsymbol{\delta}_{c,K},\boldsymbol{f},\boldsymbol{e},d,J,K,\mathscr{R}_n) + w_3n\sum_{k=0}^Kr_{k,J}\\
	& \qquad + p\left[\mathbb{I}\{FWER_I(a\mid\boldsymbol{0}_K, \boldsymbol{f},\boldsymbol{e},d,J,K,\mathscr{R}_n)>\alpha\}\left\{\frac{FWER_I(a\mid\boldsymbol{0}_K, \boldsymbol{f},\boldsymbol{e},d,J,K,\mathscr{R}_n)-\alpha}{\alpha}\right\} \right.\\
	& \left. \qquad \qquad + \mathbb{I}\{FWP(b,c\mid \boldsymbol{\delta}_{c,K},\boldsymbol{f},\boldsymbol{e},d,J,K,\mathscr{R}_n)>\beta\}\left\{\frac{FWP(b,c\mid \boldsymbol{\delta}_{c,K},\boldsymbol{f},\boldsymbol{e},d,J,K,\mathscr{R}_n)-\beta}{\beta}\right\}\right].
	\end{align*}
\end{small}
We will refer to the function in the above, which we attempt to minimise, as the objective function. Here, $p\in\mathbb{R}^+$, and the final factors therefore penalise designs not conforming to the desired operating characteristics. Following Wason and Jaki (2012), we set $p=N_{\text{fixed}}$, the sample-size required by a corresponding single-stage design.

Note that $n$ is treated in this search as a continuous parameter. This is for convenience given the majority of available software for optimisation is applicable only to continuous search spaces. Denoting the solution to the above problem by $\{n_*,(\boldsymbol{f}_*,\boldsymbol{e}_*)\}$, then, as discussed in Wason (2015), in the likely event that $n_*\notin\mathbb{N}^+$, a choice must be made. Explicitly, define $n_l = \max\{p\in\mathbb{N}_{\lfloor n_* \rfloor} : \forall(j,k)\in\mathbb{N}_J^+\times\mathbb{N}_K \ r_{k,j}p\in\mathbb{N}^+\}$ and $n_u = \min\{p\in\mathbb{N}^+\backslash\mathbb{N}_{\lfloor n_* \rfloor} : \forall(j,k)\in\mathbb{N}_J^+\times\mathbb{N}_K \ r_{k,j}p\in\mathbb{N}^+\}$. The optimal integer $n$ could then either be directly designated as $n_u$, increasing the likelihood that the FWP requirement is met. Or, the optimisation routine could be re-run with $n$ constrained in turn to $n_l$ and $n_u$ (i.e., the optimal $(\boldsymbol{f},\boldsymbol{e})$ for $n\in\{n_l,n_u\}$ should be identified), and the design with the smallest value of the objective function amongst these two solutions determined the true optimal design. This latter procedure is of most importance when $(n_u-n_l)/n_l$ is relatively large, when it is likely that the performance of the constrained optimal designs will be substantially different.

Finally, whilst the above specifies how a search can be performed, it does not enforce that a particular search algorithm be utilised. However, given the search space is likely to contain many local optima, a stochastic search routine is preferable (Wason and Jaki, 2012). Here, we utilise the GA package in R (Scrucca, 2013; Scrucca, 2016), which allows the parallel evaluation of several candidate designs. In combination with the use of the package mvtnorm (Genz et al., 2016) for evaluating the requisite multivariate normal integrals, this allow for the efficient determination of optimised designs. Of importance is that GA requires the designation of values for several control parameters. We discuss our choices for these in the Appendix. Code to replicate our results is available upon request.

\section{Example: TAILoR}\label{tailor}

To explore the effect the parameters $a$, $b$, $c$, and $d$ can have upon a designs operating characteristics, we re-consider the TAILoR (TelmisArtan and InsuLin Resistance in HIV) trial (Magirr et al., 2012). In this trial, three experimental treatments, corresponding to different doses of Telmisartan, were compared to a shared control. Therefore, we set $K=3$, and as in Magirr et al. (2012), we suppose that $\alpha=0.05$, $\beta=0.1$, $\delta=0.545$, $\delta_0=0.138$, and $\sigma_k^2=1$ for $k\in\mathbb{N}_K$. To allow utilisation of Equations~(\ref{fwer2})-(\ref{ess2}), we further suppose that $r_{k,j}=j$ for $j\in\mathbb{N}_J^+$ and $k\in\mathbb{N}_K$ (implying $n_l=\lfloor n_*\rfloor$ and $n_u=\lceil n_*\rceil$).

In Table~\ref{taba1} and Table~\ref{taba2}, we summarise the operating characteristics of determined optimal designs for $J=2$ and $a=2$, considering all possible combinations of $b$, $c$, and $d$ such that $b\le c$, when $w_1=w_2=w_3=1/3$. Corresponding results for the cases with $a=1$ and $a=3$ are given in the Appendix. In all instances, the optimal design was determined for each parameter set by running ten replicate searches from randomised starting points, and choosing amongst them the design with the smallest value of the objective function. For simplicity, the optimal integer $n$ was chosen by setting $n=n_u$.

In general, it is clear that each of the presented optimal designs controls the operating characteristics to the desired level. Perhaps surprisingly, from Table~\ref{taba1} we can see that increasing the value of $d$, with all other design parameters fixed, usually results in a design with a lower value for $FWER_I(1\mid\boldsymbol{0}_K,\boldsymbol{f},\boldsymbol{e},d,J,K,\mathscr{R}_n)$. For example, in the cases with $b=c=1$, the design with $d=1$ has $FWER_I(1\mid\boldsymbol{0}_K,\boldsymbol{f},\boldsymbol{e},d,J,K,\mathscr{R}_n)=0.233$, but with $d=3$, this drops to 0.184. Closer inspection reveals this is likely a consequence of the fact that larger values of $d$ typically result in larger values for $e_1$.

The value of $\mathbb{E}(N\mid\boldsymbol{\tau},\boldsymbol{f},\boldsymbol{e},d,J,K,\mathscr{R}_n)$ appears relatively insensitive to the choice of $d$ for fixed $b$ and $c$. In the case with $b=c=3$ for example, $\mathbb{E}(N\mid\boldsymbol{0}_K,\boldsymbol{f},\boldsymbol{e},d,J,K,\mathscr{R}_n)\in[216.3,220.9]$ no matter the choice of $d\in\mathbb{N}_3^+$. However, its value is highly dependent upon the choice for $b$ and $c$, as would be expected.

In Table~\ref{taba2}, it can be seen that the various FWP vary widely depending on the values for $b$, $c$, and $d$. In general, $FWP(p,q\mid \boldsymbol{\delta}_{r,K},\boldsymbol{f},\boldsymbol{e},d,J,K,\mathscr{R}_n)$ increases as $q$ is increased with all other parameters fixed. Larger values of $d$ also usually result in increased values for each of the presented FWPs, with all other parameters fixed. This follows the findings of Urach and Posch (2016). However, from Table~\ref{taba1}, this comes at a cost to the associated ESS. For example, when $b=1$ and $c=2$, the design with $d=1$ has $\mathbb{E}(N\mid\boldsymbol{\delta}_{3,K},\boldsymbol{f},\boldsymbol{e},d,J,K,\mathscr{R}_n)=79.6$. This increases by 49\% when $d=3$ to 118.5.

In Tables~\ref{taba3}-\ref{taba6}, we observe that the findings for the cases $a=1$ and $a=3$ are similar to those discussed here for $a=2$.

\def\arraystretch{1.2}
\begin{sidewaystable}[htbp]
	\begin{center}
		\caption{A comparison of the performance of several optimal designs when $a=2$, $K=3$, $\alpha=0.05$, $\beta=0.1$, $\delta=0.545$, $\delta_0=0.138$, $\sigma_k^2=1$ and $r_{k,j}=j$ for $j\in\mathbb{N}_J^+$ and $k\in\mathbb{N}_K$. Rejection probabilities are given to three decimal places, whilst expected sample sizes are given to one decimal place. For brevity,  $FWER_I(p\mid\boldsymbol{0}_K, \boldsymbol{f},\boldsymbol{e},d,J,K,\mathscr{R}_n)\equiv FWER_I(p)$, and $\mathbb{E}(N\mid\boldsymbol{\tau},\boldsymbol{f},\boldsymbol{e},d,J,K,\mathscr{R}_n)\equiv \mathbb{E}(N\mid\boldsymbol{\tau})$.}
		\label{taba1}
		\begin{tabular}{rrrrrrrrrrrrr}
			\hline
			$b$ & $c$ & $d$ & $n$ & $\boldsymbol{f}^\top$ & $\boldsymbol{e}^\top$ & $FWER_I(1)$ & $FWER_I(2)$ & $FWER_I(3)$ & $\mathbb{E}(N\mid\boldsymbol{0}_K)$ & $\mathbb{E}(N\mid\boldsymbol{\delta}_{1,K})$ & $\mathbb{E}(N\mid\boldsymbol{\delta}_{2,K})$ & $\mathbb{E}(N\mid\boldsymbol{\delta}_{3,K})$ \\
			\hline
			1 & 1 & 1 & 27 & $(0.08, 1.31)$ & $(1.70, 1.31)$ & 0.233 & 0.050 & 0.008 & 154.7 & 134.9 & 126.6 & 123.0 \\
			1 & 1 & 2 & 25 & $(-0.01, 1.39)$ & $(2.40, 1.39)$ & 0.183 & 0.050 & 0.010 & 156.0 & 172.1 & 166.0 & 162.4 \\
			1 & 1 & 3 & 26 & $(0.26, 1.38)$ & $(2.19, 1.38)$ & 0.184 & 0.050 & 0.011 & 150.7 & 170.6 & 170.5 & 167.7 \\
			1 & 2 & 1 & 16 & $(0.34, 1.24)$ & $(1.72, 1.24)$ & 0.237 & 0.050 & 0.008 & 85.5 & 83.8 & 81.1 & 79.6 \\
			1 & 2 & 2 & 17 & $(0.17, 1.39)$ & $(2.12, 1.39)$ & 0.184 & 0.050 & 0.009 & 100.5 & 111.3 & 110.0 & 109.2 \\
			1 & 2 & 3 & 17 & $(0.39, 1.33)$ & $(2.42, 1.33)$ & 0.184 & 0.050 & 0.011 & 95.3 & 111.2 & 115.3 & 118.5 \\
			1 & 3 & 1 & 12 & $(0.42, 1.21)$ & $(1.75, 1.21)$ & 0.238 & 0.050 & 0.008 & 63.0 & 63.9 & 62.9 & 62.3 \\
			1 & 3 & 2 & 13 & $(0.19, 1.39)$ & $(2.13, 1.39)$ & 0.184 & 0.050 & 0.009 & 76.6 & 85.1 & 85.6 & 86.1 \\
			1 & 3 & 3 & 13 & $(0.19, 1.38)$ & $(2.21, 1.38)$ & 0.184 & 0.050 & 0.011 & 76.7 & 86.4 & 88.7 & 90.4 \\
			2 & 2 & 1 & 31 & $(-0.01, 1.38)$ & $(11.00, 1.38)$ & 0.182 & 0.050 & 0.011 & 194.1 & 229.2 & 237.9 & 246.5 \\
			2 & 2 & 2 & 32 & $(-0.01, 1.39)$ & $(2.44, 1.39)$ & 0.183 & 0.050 & 0.010 & 199.9 & 219.4 & 205.3 & 197.7 \\
			2 & 2 & 3 & 31 & $(-0.09, 1.41)$ & $(2.22, 1.41)$ & 0.183 & 0.050 & 0.011 & 196.9 & 212.8 & 206.7 & 196.3 \\
			2 & 3 & 1 & 22 & $(0.05, 1.37)$ & $(12.84, 1.37)$ & 0.182 & 0.050 & 0.011 & 135.8 & 159.6 & 166.6 & 173.3 \\
			2 & 3 & 2 & 22 & $(-0.16, 1.44)$ & $(2.01, 1.44)$ & 0.184 & 0.050 & 0.009 & 141.0 & 147.9 & 138.6 & 133.0 \\
			2 & 3 & 3 & 22 & $(0.02, 1.39)$ & $(2.32, 1.39)$ & 0.183 & 0.050 & 0.011 & 136.3 & 151.8 & 152.3 & 151.3 \\
			3 & 3 & 1 & 34 & $(-0.12, 1.38)$ & $(9.00, 1.38)$ & 0.182 & 0.050 & 0.011 & 219.0 & 254.9 & 263.1 & 271.1 \\
			3 & 3 & 2 & 34 & $(-0.07, 1.38)$ & $(13.38, 1.38)$ & 0.182 & 0.050 & 0.011 & 216.3 & 253.7 & 262.4 & 271.0 \\
			3 & 3 & 3 & 35 & $(-0.08, 1.43)$ & $(2.04, 1.43)$ & 0.184 & 0.050 & 0.011 & 220.9 & 234.0 & 222.5 & 204.3 \\
			\hline
			\end{tabular}
	\end{center}
\end{sidewaystable}

\def\arraystretch{1.2}
\begin{sidewaystable}[htbp]
	\begin{center}
		\caption{A comparison of the performance of several optimal designs when $a=2$, $K=3$, $\alpha=0.05$, $\beta=0.1$, $\delta=0.545$, $\delta_0=0.138$, $\sigma_k^2=1$ and $r_{k,j}=j$ for $j\in\mathbb{N}_J^+$ and $k\in\mathbb{N}_K$. Rejection probabilities are given to three decimal places. For brevity, $FWP(p,q\mid \boldsymbol{\delta}_{r,K},\boldsymbol{f},\boldsymbol{e},d,J,K,\mathscr{R}_n)\equiv FWP(p,q\mid r)$.}
		\label{taba2}
		\begin{tabular}{rrrrrrrrr}
			\hline
			$b$ & $c$ & $d$ & $FWP(1,1\mid 1)$ & $FWP(1,2\mid 2)$ & $FWP(1,3\mid 3)$ & $FWP(2,2\mid 2)$ & $FWP(2,3\mid 3)$ & $FWP(3,3\mid 3)$ \\
			\hline
			1 & 1 & 1 & 0.909 & 0.977 & 0.991 & 0.598 & 0.760 & 0.445 \\
			1 & 1 & 2 & 0.901 & 0.969 & 0.985 & 0.826 & 0.936 & 0.615 \\
			1 & 1 & 3 & 0.902 & 0.969 & 0.986 & 0.834 & 0.936 & 0.782 \\
			1 & 2 & 1 & 0.763 & 0.903 & 0.948 & 0.425 & 0.600 & 0.272 \\
			1 & 2 & 2 & 0.785 & 0.906 & 0.946 & 0.655 & 0.826 & 0.446 \\
			1 & 2 & 3 & 0.779 & 0.901 & 0.943 & 0.656 & 0.818 & 0.576 \\
			1 & 3 & 1 & 0.672 & 0.839 & 0.904 & 0.343 & 0.513 & 0.203 \\
			1 & 3 & 2 & 0.697 & 0.844 & 0.900 & 0.542 & 0.730 & 0.361 \\
			1 & 3 & 3 & 0.699 & 0.845 & 0.901 & 0.554 & 0.731 & 0.466 \\
			2 & 2 & 1 & 0.945 & 0.986 & 0.994 & 0.903 & 0.970 & 0.869 \\
			2 & 2 & 2 & 0.949 & 0.988 & 0.995 & 0.904 & 0.973 & 0.672 \\
			2 & 2 & 3 & 0.944 & 0.986 & 0.994 & 0.902 & 0.970 & 0.869 \\
			2 & 3 & 1 & 0.870 & 0.953 & 0.976 & 0.785 & 0.908 & 0.724 \\
			2 & 3 & 2 & 0.866 & 0.952 & 0.976 & 0.769 & 0.905 & 0.526 \\
			2 & 3 & 3 & 0.869 & 0.953 & 0.976 & 0.784 & 0.907 & 0.723 \\
			3 & 3 & 1 & 0.960 & 0.991 & 0.997 & 0.928 & 0.980 & 0.903 \\
			3 & 3 & 2 & 0.959 & 0.991 & 0.997 & 0.927 & 0.980 & 0.901 \\
			3 & 3 & 3 & 0.961 & 0.991 & 0.997 & 0.930 & 0.981 & 0.905 \\
			\hline
		\end{tabular}
	\end{center}
\end{sidewaystable}

\section{Discussion}\label{disc}

In this article, we have presented methodology for the design of MAMS experiments with control of generalised error-rates and more flexible stopping rules. By developing an efficient approach to the evaluation of the statistical operating characteristics of any proposed design, we were able to determine numerous example optimised designs. Indeed, it was clear that by taking into account interchangeability between experimental arms with equal allocations, variances in response, and treatment effects, the number of multivariate normal integrals that must be considered to evaluate a designs performance could be substantially reduced (Table~\ref{tab1}). Combining this with parallel design evaluation using the GA package, optimal designs were readily identified. All calculations were performed on a Macbook Pro with an i7 processor across seven cores. Execution time for each of the 540 optimal design runs ranged from approximately ten to 15 minutes. By utilising a high performance computer, this could be reduced even further. While non-optimal boundaries, such as those based on the triangular test, could be identified extremely easily.

Computationally, the complexity of our method is comparable to that of Ghosh et al. (2017), as we also employ recently developed algorithms for the fast evaluation of multivariate normal integrals (Genz et al., 2016). Here, we view our method as preferable over their approach because of our requirement to evaluate ESSs in Section~\ref{optimal}, with the methodology in Ghosh et al. (2017) specialised to the determination of the conventional type-I FWER in the case $d=1$. It is also preferential to a simulation based approach, such as that taken by Wason and Jaki (2012), as small Monte Carlo errors can be attained without the need for large numbers of replicate simulations.

Optimal design determination is possible for many $J$ and $K$ of practical interest. Of course, this may not be the case as these parameters increase. In this instance, extending the approach of Wason (2015), which reduces the dimensionality of the optimisation problem, would likely be an effective means of retaining efficient design operating characteristics, with tractable design determination attained.

There are also limitations to the use of a stochastic search procedure for all $J$ and $K$ that must be recognised. Firstly, there is no guarantee that the search will converge to the global optimial design. Evaluation with multiple randomised starting values is therefore typically important (Wason and Jaki, 2012). Likewise, it is vital to make appropriate choices for the algorithms control parameters, as discussed in the Appendix. Moreover, for these reasons, investigating the performance of a dynamic programming approach, similar to the method of Barber and Jennison (2002) for two-arm group sequential trials, could be helpful. It is not currently clear how well this approach would perform in the case of multiple arms.

Whilst MAMS methodology is applicable in many experimental design settings, as was discussed, much motivation for its development has come from the desire to improve the efficiency of clinical research. As was noted by Ghosh et al. (2017), it may be useful in the design of platform trials, for example. However, we believe the methodology presented here would be of greatest use in the design of phase II clinical trials, the non-confirmatory nature of which means strong control of the conventional type-I FWER may not always be required (Wason et al., 2014), and in which interest in randomised designs is increasing but is hampered by conventional sample size requirements. Considering the performance of the designs presented in Tables~\ref{taba1}-\ref{taba6}, it is clear controlling generalised error-rates may be able to ameliorate this issue. For example, the optimal 1111-MAMS design required a group size of 40. In contrast, the 1131-MAMS and 2132-MAMS designs required group sizes of only 24 and 13 respectively.

Finally, following the arguments of Jaki and Magirr (2013), our results presented here are applicable to a variety of other types of outcome variable. However, as noted in Magirr et al. (2012), application to a scenario with normally distributed outcome variables of unknown variance is not straightforward. In this case, for multi-arm trials, a quantile substitution procedure has been considered for boundary specification (Wason et al., 2015). Additionally, an approach based on Monte Carlo simulation, for direct determination of optimal group sizes and stopping boundaries, has also been presented (Grayling et al., 2017).

\section{Appendix A}\label{appA}

\subsection{Sample spaces}

In Section~\ref{desdet}, we introduced the sample space for the random pair $(\boldsymbol{\Psi},\boldsymbol{\Omega})$, which we defined by $\Xi(\boldsymbol{f},\boldsymbol{e}|d,J,K)$. Precisely, this is
\begin{equation*}\label{xiEF}
\begin{split}
\Xi(\boldsymbol{f},\boldsymbol{e}|d,J,K) &=\left[(\boldsymbol{\psi},\boldsymbol{\omega})\in(\mathbb{N}_1)^K\times(\mathbb{N}_J^+)^K : \forall j\in\mathbb{N}_J^+\  \sum_{k=1}^{K}\mathbb{I}(\psi_k=1)\mathbb{I}(\omega_k \le j) \ge d \Longrightarrow \sum_{k=1}^{K}\mathbb{I}(\omega_k > j) = 0, \right.\\
&\qquad \qquad \left. \forall k\in\mathbb{N}_K^+\ (\psi_k=1,\omega_k=j)\Longrightarrow e_j\neq\infty,\right.\\ & \qquad \qquad \qquad \left. \forall k\in\mathbb{N}_K^+\ (\psi_k=0,\omega_k<J)\Longrightarrow (f_j\neq-\infty)\ \cup\right.\\
&\qquad\qquad\qquad\qquad \left. \left\{\sum_{l\neq k}\mathbb{I}(\omega_l>j)=0, \sum_{l\neq k}\mathbb{I}(\psi_l=1)\mathbb{I}(\omega_l=j)>0 \right\} \right].
\end{split}
\end{equation*}
The conditions here account for the fact that the trial terminates if any $d$ null hypotheses have been rejected, that one can only reject a null hypothesis at an analysis with a non-infinite rejection boundary, and that the pair $(\psi_k,\omega_k)=(0,j)$ is only possible if the acceptance boundary at stage $j$ is non-infinite or this is the analysis at which the trial was terminated due to other hypotheses having been rejected.

We now proceed to derive formulae for $|\Xi(\boldsymbol{f},\boldsymbol{e}|d,J,K)|$ in the case that $(\boldsymbol{f},\boldsymbol{e})\in\mathscr{B}'$. To this end, define
\[ \Xi_j(\boldsymbol{f},\boldsymbol{e}|d,J,K) = \left\{ (\boldsymbol{\psi},\boldsymbol{\omega})\in\Xi(\boldsymbol{f},\boldsymbol{e}|d,J,K) : \max_k \omega_k=j \right\}, \qquad j\in\mathbb{N}_J^+. \]
Then,
\[ \Xi(\boldsymbol{f},\boldsymbol{e}|d,J,K)=\bigcup_{j=1}^J \Xi_j(\boldsymbol{f},\boldsymbol{e}|d,J,K), \]
and
\[ |\Xi(\boldsymbol{f},\boldsymbol{e}|d,J,K)|=\sum_{j=1}^J|\Xi_j(\boldsymbol{f},\boldsymbol{e}|d,J,K)|.\]
We have, when $(\boldsymbol{f},\boldsymbol{e})\in\mathscr{B}'$
\begin{align*}
|\Xi_1(\boldsymbol{f},\boldsymbol{e}|d,J,K)|&=2^K,\\
|\Xi_2(\boldsymbol{f},\boldsymbol{e}|d,J,K)|&=\sum_{k=0}^{K-1}\sum_{r=0}^{\min(k,d-1)}\binom{k}{r} \binom{k}{k-r}2^{K-k},\\
|\Xi_j(\boldsymbol{f},\boldsymbol{e}|d,J,K)|&=\sum_{k_1=0}^{K-1}\ \sum_{k_2=0}^{K-k_1-1}\dots\sum_{k_{j-1}=0}^{K-\sum_{l=1}^{j-2}k_l-1}\ \sum_{r_1=0}^{\min(k_1,d-1)}\ \sum_{r_2=0}^{\min(k_2,d-r_1-1)}\ \dots\\
& \qquad \qquad \dots \sum_{r_{j-1}=0}^{\min(k_{j-1},d-\sum_{l=1}^{j-2}r_l-1)}\ \prod_{l=1}^{j-1}\binom{k_l}{r_l}\binom{k_l}{k_l-r_l}2^{K-\sum_{l=1}^{j-1}k_l}, \qquad j\in\mathbb{N}_J\backslash\mathbb{N}_2.
\end{align*}
These formulae arise from considering the possible combinations of experimental treatments that could be dropped from the trial at each stage, along with which of these had their null hypotheses rejected.

Next, we perform similar derivations for $\Xi'(\boldsymbol{f},\boldsymbol{e},\boldsymbol{\tau}|d,J,K)$, again in the case $(\boldsymbol{f},\boldsymbol{e})\in\mathscr{B}'$. First however, we motivate our definition of $\Xi'(\boldsymbol{f},\boldsymbol{e},\boldsymbol{\tau}|d,J,K)$. Note that $2\omega_{k}-\mathbb{I}(\psi_{k}=0)\in\mathbb{N}_{2J}^+$ provides a unique value that describes the outcome for treatment $k$. That is, it maps the 2-dimensional description given by $(\psi_k,\omega_k)$ to a 1-dimensional description. This allows it to be used for ordering the outcomes for interchangeable interventions. Whilst the use of the ordering itself is simply to prevent repeated counting of events with equal probability. A simple example of this would be in assessing the sum of the scores on three rolled dice. If the dice are not distinguishable in any way, we may order the outcomes to avoid, for example, consideration of the events $\{1,5,3\}$, $\{3,1,5\}$, $\{3,5,1\}$, $\{5,1,3\}$ and $\{5,3,1\}$, as well as the ordered event $\{1,3,5\}$.

Now, following the above, set
\[ \Xi'_j(\boldsymbol{f},\boldsymbol{e},\boldsymbol{\tau}|d,J,K) = \left\{ (\boldsymbol{\psi},\boldsymbol{\omega})\in\Xi'(\boldsymbol{f},\boldsymbol{e},\boldsymbol{\tau}|d,J,K) : \max_k \omega_k=j \right\}, \qquad j\in\mathbb{N}_J^+. \]
Then
\[ \Xi'(\boldsymbol{f},\boldsymbol{e},\boldsymbol{\tau}|d,J,K)=\bigcup_{j=1}^J \Xi'_j(\boldsymbol{f},\boldsymbol{e},\boldsymbol{\tau}|d,J,K), \]
and
\[ |\Xi'(\boldsymbol{f},\boldsymbol{e},\boldsymbol{\tau}|d,J,K)|=\sum_{j=1}^J|\Xi'_j(\boldsymbol{f},\boldsymbol{e},\boldsymbol{\tau}|d,J,K)|.\]
For $\boldsymbol{\tau}=\boldsymbol{0}_K$
\begin{align*}
|\Xi'_1(\boldsymbol{f},\boldsymbol{e},\boldsymbol{0}_K|d,J,K)|&=K+1,\\
|\Xi'_2(\boldsymbol{f},\boldsymbol{e},\boldsymbol{0}_K|d,J,K)|&=\sum_{k=0}^{K-1}\sum_{r=0}^{\min(k,d-1)}(K-k+1),\\
|\Xi'_j(\boldsymbol{f},\boldsymbol{e},\boldsymbol{0}_K|d,J,K)|&=\sum_{k_1=0}^{K-1}\ \sum_{k_2=0}^{K-k_1-1}\dots\sum_{k_{j-1}=0}^{K-\sum_{l=1}^{j-2}k_l-1}\ \sum_{r_1=0}^{\min(k_1,d-1)}\ \sum_{r_2=0}^{\min(k_2,d-r_1-1)}\ \dots\\
& \qquad \dots\sum_{r_{j-1}=0}^{\min(k_{j-1},d-\sum_{l=1}^{j-2}r_l-1)}\ \left(K - \sum_{l=1}^{j-1}k_l + 1\right), \qquad j\in\mathbb{N}_J\backslash\mathbb{N}_2.
\end{align*}	
Moreover, when $\boldsymbol{\tau}=\boldsymbol{\delta}_{c,K}$
\begin{align*}
|\Xi'_1(\boldsymbol{f},\boldsymbol{e},\boldsymbol{\delta}_{c,K}|d,J,K)|&=(c+1)(K-c+1),\\
|\Xi'_2(\boldsymbol{f},\boldsymbol{e},\boldsymbol{\delta}_{c,K}|d,J,K)|&=\sum_{k_E=0}^{c}\sum_{k_F=0}^{\min(K-k_E-1,K-c)}\sum_{r_E=0}^{\min(k_E,d-1)}\sum_{r_F=0}^{\min(k_F,d-r_E-1)}(c-k_E+1)(K-c-k_F+1),\\
|\Xi'_j(\boldsymbol{f},\boldsymbol{e},\boldsymbol{\delta}_{c,K}|d,J,K)|&=\sum_{k_{E1}=0}^{c}\sum_{k_{F1}=0}^{\min(K-k_{E1}-1,K-c)}\dots\sum_{k_{Ej-1}=0}^{c-\sum_{l=1}^{j-2}k_El}\ \sum_{k_{Fj-1}=0}^{\min(K-\sum_{l=1}^{j-2}k_{El}-1,K-\sum_{l=1}^{j-2}k_{Fl}-c)}\\
& \qquad \sum_{r_{E1}=0}^{\min(k_{E1},d-1)}\ \sum_{r_{F1}=0}^{\min(k_{F1},d-r_{E1}-1)}\dots\sum_{r_{Ej-1}=0}^{\min(k_{Ej-1},d-\sum_{l=1}^{j-2}(r_{El}+r_{Fl})-1)}\\
& \qquad \qquad \sum_{r_{Fj-1}=0}^{\min(k_{Fj-1},d-\sum_{l=1}^{j-2}(r_{El}+r_{Fl})-r_{Ej-1}-1)}\left(c-\sum_{l=1}^{j-1}k_{El}+1\right)\\
& \qquad \qquad \qquad \qquad \qquad \qquad\qquad \qquad\times \left(K-c-\sum_{l=1}^{j-1}k_{Fl}+1\right), \qquad j\in\mathbb{N}_J\backslash\mathbb{N}_2.
\end{align*}
These formulae arise by adapting those given above for $\Xi(\boldsymbol{f},\boldsymbol{e}|d,J,K)$ to take in to account the ordering placed on the outcomes for the experimental treatment arms.	

Finally
\begin{align*}
\text{d}(\boldsymbol{\psi},\boldsymbol{\omega}|\boldsymbol{0}_K,\boldsymbol{f},\boldsymbol{e},d,J,K) &= |\{2\omega_1-\mathbb{I}(\psi_1=0),\dots,2\omega_K-\mathbb{I}(\psi_K=0)\}_{\neq}|!,\\
\text{d}(\boldsymbol{\psi},\boldsymbol{\omega}|\boldsymbol{\delta}_{c,K},\boldsymbol{f},\boldsymbol{e},d,J,K) &= |\{2\omega_1-\mathbb{I}(\psi_1=0),\dots,2\omega_c-\mathbb{I}(\psi_c=0)\}_{\neq}|!\\
& \qquad \times |\{2\omega_{c+1}-\mathbb{I}(\psi_{c+1}=0),\dots,2\omega_K-\mathbb{I}(\psi_K=0)\}_{\neq}|!,
\end{align*}
where $S_{\neq}$ denotes the set of unique elements belonging to the set $S$.

\subsection{Modified triangular test}

Suppose that values for $a$, $b$, $c$, $d$, $J$, $K$, $\delta$, $\delta_0$, $\alpha$, and $\beta$, have all been designated, and that $r_{k,j}=j$ for $j\in\mathbb{N}_J^+$ and $k\in\mathbb{N}_K$. Additionally, assume that $\sigma_0^2$ and $\sigma_1^2$ have been specified, and $\sigma_{1}^2=\dots=\sigma_K^2$. The boundaries and group size of a modified triangular test can then be determined by finding the solution to the following optimisation problem
\begin{align*}
\argmin_{(\alpha',\beta') \in (0,1)^2} &[FWER_I\{a\mid\boldsymbol{0}_K, \boldsymbol{f}(\alpha',\beta'),\boldsymbol{e}(\alpha',\beta'),d,J,K,\mathscr{R}_{n(\alpha',\beta')}\} - \alpha]^2\\ & \qquad + [FWP\{b,c\mid \boldsymbol{\delta}_{c,K},\boldsymbol{f}(\alpha',\beta'),\boldsymbol{e}(\alpha',\beta'),d,J,K,\mathscr{R}_{n(\alpha',\beta')}\} - \{1-\beta\}]^2,
\end{align*}\label{triang}
where
\begin{align*}
\tilde{\delta} &= \frac{2\Phi^{-1}(1-\alpha')}{\Phi^{-1}(1-\alpha')+\Phi^{-1}(1-\beta')}\delta,\\
\boldsymbol{I} &= (I_1,\dots,I_J)^\top = (1,\dots,J)^\top\left\{\sqrt{\frac{4(0.583)^2}{J} + 8\log\left(\frac{1}{2\alpha'}\right)} - \frac{2(0.583)}{\sqrt{J}}\right\}^2\frac{1}{J\tilde{\delta}}\\
\boldsymbol{f}(\alpha',\beta') &= -\frac{2}{\tilde{\delta}}\log\left(\frac{1}{2\alpha'}\right) + 0.583\sqrt{\frac{I_J}{J}}+\frac{3\tilde{\delta}}{4}\frac{I_J}{J}(1,\dots,J)^\top\circ(I_1^{-1/2},\dots,I_J^{-1/2})^\top,\\
\boldsymbol{e}(\alpha',\beta') &= \frac{2}{\tilde{\delta}}\log\left(\frac{1}{2\alpha'}\right) - 0.583\sqrt{\frac{I_J}{J}}+\frac{\tilde{\delta}}{4}\frac{I_J}{J}(1,\dots,J)^\top\circ(I_1^{-1/2},\dots,I_J^{-1/2})^\top,\\
n(\alpha',\beta') &= (\sigma_0^2+\sigma_1^2)I_1^{-1}.
\end{align*}
Explicitly, a two-dimensional numerical optimisation procedure can be used to find the $(\alpha',\beta')\in(0,1)^2$ such that the objective function above is equal to zero. Then, a group size of $n(\alpha',\beta')$, accompanied by stopping boundaries $\boldsymbol{e}(\alpha',\beta')$ and $\boldsymbol{f}(\alpha',\beta')$, provides a triangular test design with the desired statistical operating characteristics.

\subsection{GA control parameters}

The function \texttt{ga}, from the package GA, utilised in searching for optimised generalised MAMS designs in this article, allows a large number of inputs controlling the algorithms conduct to be specified. In theory, each of these could be tailored for improved reliability and speed in converging toward the global optimiser. In reality, this is not necessary for the majority of inputs, as default values routinely lead to designs with impressive operating characteristics. However, of particular importance are the inputs \texttt{popSize} and \texttt{maxiter}. The former designates the number of candidate designs considered in each iteration of the algorithm, and the latter the maximum number of iterations to allow. Setting either value too low would often result in non-convergence. Values too high though lead to undesirable run times.

To explore possible sensible choices for these values, we examined convergence for several of design scenarios. Explicitly, we set $\alpha=0.05$, $\beta=0.1$, $\delta=0.545$, $\delta_0=0.138$, $\sigma_k^2=1$ and $r_{k,j}=j$ for $j\in\mathbb{N}_J^+$ and $k\in\mathbb{N}_K$, as in Section~\ref{tailor}, with $a=b=c=d=1$. Then, we considered three scenarios defined by $(J,K)\in\{(2,2),(2,3),(3,2)\}^2$. Ten replicate searches were performed for each scenario with random starting values, with \texttt{popSize} equal to 50, 100, or 200, and \texttt{maxiter} equal to 1000. The value of the objective function of the best design at each iteration, for each replicate, was recorded. These are plotted as a function of the iteration number in Figure 1. It is clear that in each case little reduction in the objective function for the best identified design typically occurs beyond iteration 500. In addition, increasing the value of \texttt{popSize} leads to a reduction in the number of iterations typically required to converge to the final returned solution. Therefore, it would seem preferable to in general perform additional replicate searches from random starting points, rather than to consider increasing \texttt{popSize} or \texttt{maxiter}.

Based on these findings, in the examples presented in this article, we set \texttt{popSize} as 100, and \texttt{maxiter} as 500, and performed ten replicate searches. This appeared adequate for identifying the (approximately) optimal design. Nonetheless, caution should be taken in assuming such values will universally be acceptable.

\begin{figure}[htb]
	\begin{center}\label{fig1}
		\includegraphics[width=14cm]{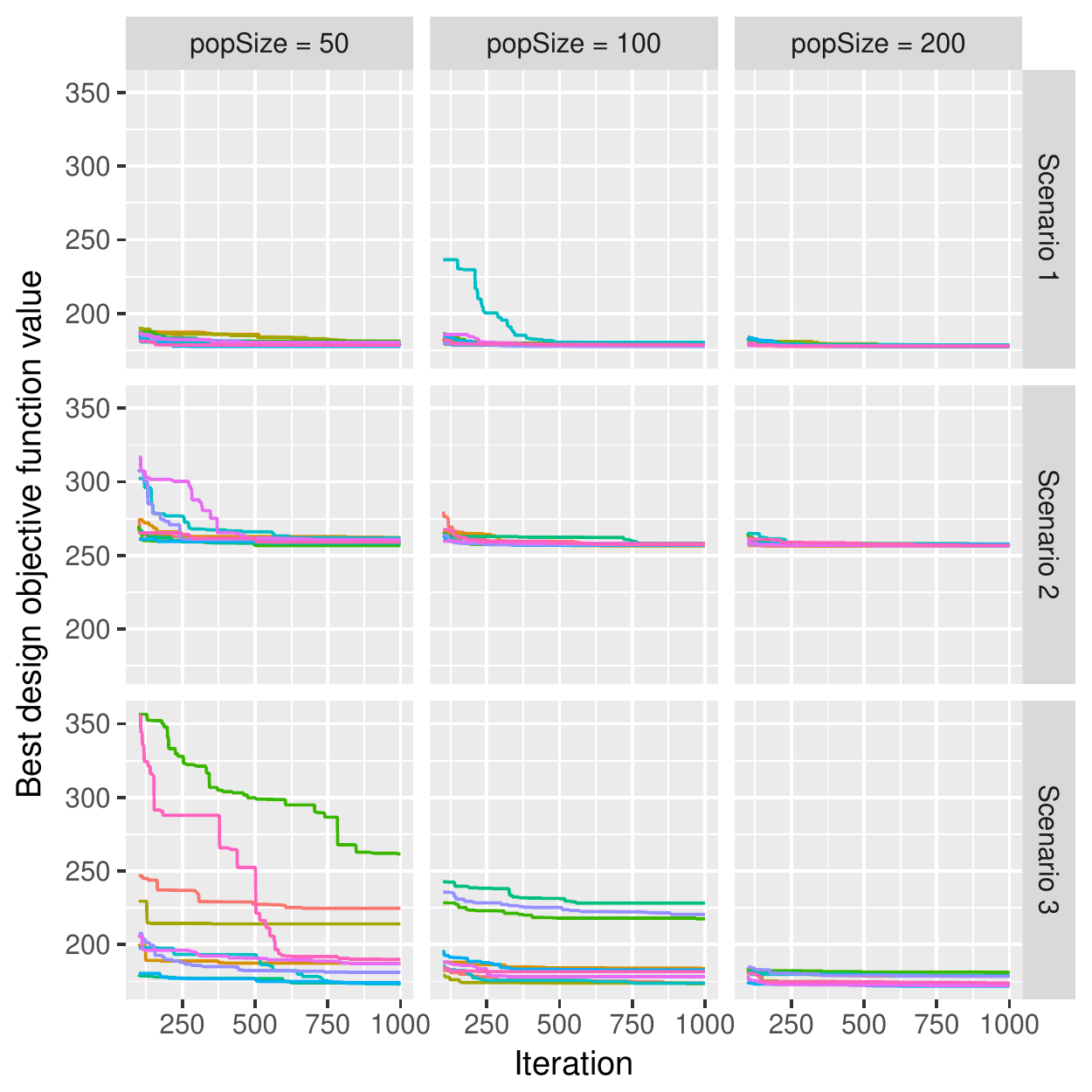}
		\caption{The convergence of the GA algorithm is explored for several scenarios, and several values of \texttt{popSize}. The value of the objective function of the best design at each iteration of ten replicate executions of the algorithm, for each parameter set, are shown. In Scenario 1 $K=2$ and $J=2$, in Scenario 2 $K=3$ and $J=2$, and in Scenario 3 $K=2$ and $J=3$.}
	\end{center}
\end{figure}

\subsection{Additional optimal designs}

Identified optimal designs in the cases with $a=1$ and $a=3$ are presented, with all other parameters specified as in Section~\ref{tailor}.

\def\arraystretch{1.2}
\begin{sidewaystable}[htbp]
	\begin{center}
		\caption{A comparison of the performance of several optimal designs when $a=1$, $K=3$, $\alpha=0.05$, $\beta=0.1$, $\delta=0.545$, $\delta_0=0.138$, $\sigma_k^2=1$ and $r_{k,j}=j$ for $j\in\mathbb{N}_J^+$ and $k\in\mathbb{N}_K$. Rejection probabilities are given to three decimal places, whilst expected sample sizes are given to one decimal place. For brevity,  $FWER_I(p\mid\boldsymbol{0}_K, \boldsymbol{f},\boldsymbol{e},d,J,K,\mathscr{R}_n)\equiv FWER_I(p)$, and $\mathbb{E}(N\mid\boldsymbol{\tau},\boldsymbol{f},\boldsymbol{e},d,J,K,\mathscr{R}_n)\equiv \mathbb{E}(N\mid\boldsymbol{\tau})$.}
		\label{taba3}
		\begin{tabular}{rrrrrrrrrrrrr}
			\hline
			$b$ & $c$ & $d$ & $n$ & $\boldsymbol{f}^\top$ & $\boldsymbol{e}^\top$ & $FWER_I(1)$ & $FWER_I(2)$ & $FWER_I(3)$ & $\mathbb{E}(N\mid\boldsymbol{0}_K)$ & $\mathbb{E}(N\mid\boldsymbol{\delta}_{1,K})$ & $\mathbb{E}(N\mid\boldsymbol{\delta}_{2,K})$ & $\mathbb{E}(N\mid\boldsymbol{\delta}_{3,K})$ \\
			\hline
			1 & 1 & 1 & 40 & $(0.52, 2.06)$ & $(2.85, 2.06)$ & 0.050 & 0.006 & 0.001 & 216.3 & 232.3 & 225.6 & 223.6 \\
			1 & 1 & 2 & 39 & $(0.53, 2.05)$ & $(3.01, 2.05)$ & 0.050 & 0.008 & 0.001 & 211.2 & 260.2 & 259.5 & 261.1 \\
			1 & 1 & 3 & 39 & $(0.58, 2.04)$ & $(3.11, 2.04)$ & 0.050 & 0.008 & 0.001 & 208.6 & 260.8 & 270.3 & 277.1 \\
			1 & 2 & 1 & 29 & $(0.84, 2.03)$ & $(2.86, 2.03)$ & 0.050 & 0.006 & 0.001 & 144.2 & 168.9 & 171.6 & 174.7 \\
			1 & 2 & 2 & 28 & $(0.63, 2.05)$ & $(2.90, 2.05)$ & 0.050 & 0.008 & 0.001 & 147.4 & 182.3 & 188.0 & 193.1 \\
			1 & 2 & 3 & 28 & $(0.67, 2.04)$ & $(2.98, 2.04)$ & 0.050 & 0.008 & 0.001 & 146.0 & 182.6 & 192.5 & 200.8 \\
			1 & 3 & 1 & 24 & $(0.72, 2.08)$ & $(2.61, 2.08)$ & 0.050 & 0.006 & 0.001 & 122.4 & 138.4 & 139.4 & 141.0 \\
			1 & 3 & 2 & 24 & $(0.84, 2.02)$ & $(2.91, 2.02)$ & 0.050 & 0.008 & 0.001 & 119.5 & 149.6 & 157.9 & 164.8 \\
			1 & 3 & 3 & 23 & $(0.47, 2.05)$ & $(3.16, 2.05)$ & 0.050 & 0.008 & 0.001 & 126.8 & 155.2 & 163.6 & 171.1 \\
			2 & 2 & 1 & 46 & $(0.52, 2.04)$ & $(14.80, 2.04)$ & 0.050 & 0.008 & 0.001 & 250.5 & 325.8 & 345.9 & 365.5 \\
			2 & 2 & 2 & 48 & $(0.62, 2.07)$ & $(2.69, 2.07)$ & 0.050 & 0.007 & 0.001 & 253.3 & 304.0 & 283.3 & 274.7 \\
			2 & 2 & 3 & 47 & $(0.57, 2.07)$ & $(2.76, 2.07)$ & 0.050 & 0.007 & 0.001 & 251.6 & 305.1 & 305.4 & 299.4 \\
			2 & 3 & 1 & 35 & $(0.66, 2.03)$ & $(9.32, 2.03)$ & 0.050 & 0.008 & 0.001 & 183.2 & 239.4 & 257.2 & 274.4 \\
			2 & 3 & 2 & 35 & $(0.35, 2.11)$ & $(2.61, 2.11)$ & 0.050 & 0.007 & 0.001 & 198.8 & 232.0 & 222.9 & 218.8 \\
			2 & 3 & 3 & 35 & $(0.46, 2.10)$ & $(2.61, 2.10)$ & 0.050 & 0.007 & 0.001 & 193.0 & 230.2 & 232.9 & 232.2 \\
			3 & 3 & 1 & 50 & $(0.37, 2.05)$ & $(15.73, 2.05)$ & 0.050 & 0.008 & 0.001 & 283.4 & 361.6 & 380.3 & 398.6 \\
			3 & 3 & 2 & 50 & $(0.41, 2.05)$ & $(11.67, 2.05)$ & 0.050 & 0.008 & 0.001 & 279.9 & 359.9 & 379.3 & 398.4 \\
			3 & 3 & 3 & 51 & $(0.29, 2.14)$ & $(2.49, 2.14)$ & 0.050 & 0.007 & 0.001 & 294.7 & 334.9 & 321.2 & 296.5 \\
			\hline
		\end{tabular}
	\end{center}
\end{sidewaystable}

\def\arraystretch{1.2}
\begin{sidewaystable}[htbp]
	\begin{center}
		\caption{A comparison of the performance of several optimal designs when $a=3$, $K=3$, $\alpha=0.05$, $\beta=0.1$, $\delta=0.545$, $\delta_0=0.138$, $\sigma_k^2=1$ and $r_{k,j}=j$ for $j\in\mathbb{N}_J^+$ and $k\in\mathbb{N}_K$. Rejection probabilities are given to three decimal places, whilst expected sample sizes are given to one decimal place. For brevity,  $FWER_I(p\mid\boldsymbol{0}_K, \boldsymbol{f},\boldsymbol{e},d,J,K,\mathscr{R}_n)\equiv FWER_I(p)$, and $\mathbb{E}(N\mid\boldsymbol{\tau},\boldsymbol{f},\boldsymbol{e},d,J,K,\mathscr{R}_n)\equiv \mathbb{E}(N\mid\boldsymbol{\tau})$.}
		\label{taba4}
		\begin{tabular}{rrrrrrrrrrrrr}
			\hline
			$b$ & $c$ & $d$ & $n$ & $\boldsymbol{f}^\top$ & $\boldsymbol{e}^\top$ & $FWER_I(1)$ & $FWER_I(2)$ & $FWER_I(3)$ & $\mathbb{E}(N\mid\boldsymbol{0}_K)$ & $\mathbb{E}(N\mid\boldsymbol{\delta}_{1,K})$ & $\mathbb{E}(N\mid\boldsymbol{\delta}_{2,K})$ & $\mathbb{E}(N\mid\boldsymbol{\delta}_{3,K})$ \\
			\hline
			1 & 1 & 1 & 18 & $(-0.49, 0.59)$ & $(1.00, 0.59)$ & 0.545 & 0.193 & 0.050 & 103.3 & 84.1 & 79.3 & 77.2 \\
			1 & 1 & 2 & 15 & $(-1.09, 0.83)$ & $(1.18, 0.83)$ & 0.455 & 0.204 & 0.050 & 105.9 & 94.9 & 84.3 & 78.1 \\
			1 & 1 & 3 & 16 & $(-0.20, 0.79)$ & $(2.04, 0.79)$ & 0.393 & 0.163 & 0.050 & 103.9 & 111.3 & 110.9 & 109.5 \\
			1 & 2 & 1 & 7 & $(-0.57, 0.56)$ & $(1.07, 0.56)$ & 0.548 & 0.193 & 0.049 & 41.5 & 37.4 & 35.7 & 34.6 \\
			1 & 2 & 2 & 8 & $(-0.59, 0.82)$ & $(1.15, 0.82)$ & 0.457 & 0.206 & 0.050 & 52.4 & 50.3 & 47.7 & 45.7 \\
			1 & 2 & 3 & 9 & $(-0.33, 0.84)$ & $(1.65, 0.84)$ & 0.394 & 0.163 & 0.050 & 59.2 & 61.6 & 61.4 & 60.7 \\
			1 & 3 & 1 & 4 & $(-0.39, 0.52)$ & $(1.06, 0.52)$ & 0.552 & 0.194 & 0.049 & 22.9 & 21.7 & 21.0 & 20.6 \\
			1 & 3 & 2 & 6 & $(0.14, 0.59)$ & $(1.24, 0.59)$ & 0.465 & 0.212 & 0.050 & 33.3 & 34.2 & 33.9 & 33.5 \\
			1 & 3 & 3 & 7 & $(-0.16, 0.80)$ & $(1.74, 0.80)$ & 0.395 & 0.163 & 0.050 & 44.6 & 47.2 & 47.7 & 47.9 \\
			2 & 2 & 1 & 20 & $(-0.65, 0.83)$ & $(14.85, 0.83)$ & 0.390 & 0.162 & 0.050 & 143.1 & 154.2 & 156.9 & 159.5 \\
			2 & 2 & 2 & 23 & $(-0.40, 0.76)$ & $(1.31, 0.76)$ & 0.451 & 0.202 & 0.050 & 148.3 & 138.0 & 120.0 & 111.8 \\
			2 & 2 & 3 & 20 & $(-0.70, 0.84)$ & $(1.93, 0.84)$ & 0.391 & 0.162 & 0.050 & 142.5 & 143.2 & 138.1 & 130.8 \\
			2 & 3 & 1 & 13 & $(-0.36, 0.81)$ & $(9.94, 0.81)$ & 0.391 & 0.162 & 0.050 & 88.3 & 97.1 & 99.8 & 102.4 \\
			2 & 3 & 2 & 12 & $(-0.60, 0.85)$ & $(1.09, 0.85)$ & 0.458 & 0.206 & 0.050 & 78.1 & 72.9 & 66.9 & 63.1 \\
			2 & 3 & 3 & 13 & $(-0.43, 0.85)$ & $(1.68, 0.85)$ & 0.393 & 0.162 & 0.050 & 87.5 & 89.9 & 88.1 & 85.5 \\
			3 & 3 & 1 & 23 & $(-0.62, 0.83)$ & $(9.85, 0.83)$ & 0.390 & 0.162 & 0.050 & 163.6 & 177.4 & 180.5 & 183.5 \\
			3 & 3 & 2 & 23 & $(-0.63, 0.83)$ & $(16.75, 0.83)$ & 0.390 & 0.162 & 0.050 & 163.8 & 177.5 & 180.5 & 183.5 \\
			3 & 3 & 3 & 23 & $(-0.76, 0.88)$ & $(1.56, 0.88)$ & 0.394 & 0.163 & 0.050 & 162.8 & 157.0 & 146.9 & 132.9 \\
			\hline
		\end{tabular}
	\end{center}
\end{sidewaystable}

\def\arraystretch{1.2}
\begin{sidewaystable}[htbp]
	\begin{center}
		\caption{A comparison of the performance of several optimal designs when $a=1$, $K=3$, $\alpha=0.05$, $\beta=0.1$, $\delta=0.545$, $\delta_0=0.138$, $\sigma_k^2=1$ and $r_{k,j}=j$ for $j\in\mathbb{N}_J^+$ and $k\in\mathbb{N}_K$. Rejection probabilities are given to three decimal places. For brevity, $FWP(p,q\mid \boldsymbol{\delta}_{r,K},\boldsymbol{f},\boldsymbol{e},d,J,K,\mathscr{R}_n)\equiv FWP(p,q\mid r)$.}
		\label{taba5}
		\begin{tabular}{rrrrrrrrr}
			\hline
			$b$ & $c$ & $d$ & $FWP(1,1\mid 1)$ & $FWP(1,2\mid 2)$ & $FWP(1,3\mid 3)$ & $FWP(2,2\mid 2)$ & $FWP(2,3\mid 3)$ & $FWP(3,3\mid 3)$ \\
			\hline
			1 & 1 & 1 & 0.904 & 0.971 & 0.987 & 0.561 & 0.680 & 0.388 \\
			1 & 1 & 2 & 0.901 & 0.969 & 0.985 & 0.832 & 0.936 & 0.636 \\
			1 & 1 & 3 & 0.900 & 0.969 & 0.985 & 0.832 & 0.935 & 0.781 \\
			1 & 2 & 1 & 0.777 & 0.902 & 0.944 & 0.455 & 0.598 & 0.298 \\
			1 & 2 & 2 & 0.779 & 0.900 & 0.942 & 0.654 & 0.817 & 0.479 \\
			1 & 2 & 3 & 0.778 & 0.900 & 0.941 & 0.656 & 0.817 & 0.575 \\
			1 & 3 & 1 & 0.694 & 0.844 & 0.902 & 0.366 & 0.518 & 0.220 \\
			1 & 3 & 2 & 0.701 & 0.846 & 0.903 & 0.554 & 0.734 & 0.395 \\
			1 & 3 & 3 & 0.700 & 0.844 & 0.901 & 0.556 & 0.732 & 0.468 \\
			2 & 2 & 1 & 0.943 & 0.986 & 0.994 & 0.900 & 0.969 & 0.866 \\
			2 & 2 & 2 & 0.947 & 0.987 & 0.995 & 0.903 & 0.971 & 0.636 \\
			2 & 2 & 3 & 0.944 & 0.986 & 0.994 & 0.902 & 0.970 & 0.868 \\
			2 & 3 & 1 & 0.865 & 0.951 & 0.975 & 0.778 & 0.903 & 0.715 \\
			2 & 3 & 2 & 0.864 & 0.951 & 0.975 & 0.773 & 0.902 & 0.535 \\
			2 & 3 & 3 & 0.863 & 0.950 & 0.974 & 0.776 & 0.902 & 0.713 \\
			3 & 3 & 1 & 0.960 & 0.991 & 0.997 & 0.929 & 0.981 & 0.904 \\
			3 & 3 & 2 & 0.960 & 0.991 & 0.997 & 0.928 & 0.980 & 0.902 \\
			3 & 3 & 3 & 0.959 & 0.991 & 0.997 & 0.927 & 0.980 & 0.901 \\
			\hline
		\end{tabular}
	\end{center}
\end{sidewaystable}

\def\arraystretch{1.2}
\begin{sidewaystable}[htbp]
	\begin{center}
		\caption{A comparison of the performance of several optimal designs when $a=3$, $K=3$, $\alpha=0.05$, $\beta=0.1$, $\delta=0.545$, $\delta_0=0.138$, $\sigma_k^2=1$ and $r_{k,j}=j$ for $j\in\mathbb{N}_J^+$ and $k\in\mathbb{N}_K$. Rejection probabilities are given to three decimal places. For brevity, $FWP(p,q\mid \boldsymbol{\delta}_{r,K},\boldsymbol{f},\boldsymbol{e},d,J,K,\mathscr{R}_n)\equiv FWP(p,q\mid r)$.}
		\label{taba6}
		\begin{tabular}{rrrrrrrrr}
			\hline
			$b$ & $c$ & $d$ & $FWP(1,1\mid 1)$ & $FWP(1,2\mid 2)$ & $FWP(1,3\mid 3)$ & $FWP(2,2\mid 2)$ & $FWP(2,3\mid 3)$ & $FWP(3,3\mid 3)$ \\
			\hline
			1 & 1 & 1 & 0.917 & 0.985 & 0.996 & 0.684 & 0.837 & 0.560 \\
			1 & 1 & 2 & 0.905 & 0.973 & 0.987 & 0.815 & 0.943 & 0.579 \\
			1 & 1 & 3 & 0.900 & 0.968 & 0.985 & 0.832 & 0.935 & 0.780 \\
			1 & 2 & 1 & 0.748 & 0.905 & 0.957 & 0.450 & 0.634 & 0.308 \\
			1 & 2 & 2 & 0.774 & 0.903 & 0.943 & 0.623 & 0.821 & 0.412 \\
			1 & 2 & 3 & 0.777 & 0.900 & 0.941 & 0.655 & 0.816 & 0.575 \\
			1 & 3 & 1 & 0.637 & 0.828 & 0.907 & 0.349 & 0.531 & 0.221 \\
			1 & 3 & 2 & 0.693 & 0.847 & 0.903 & 0.518 & 0.734 & 0.319 \\
			1 & 3 & 3 & 0.715 & 0.856 & 0.910 & 0.574 & 0.749 & 0.486 \\
			2 & 2 & 1 & 0.943 & 0.986 & 0.994 & 0.901 & 0.969 & 0.867 \\
			2 & 2 & 2 & 0.961 & 0.993 & 0.997 & 0.912 & 0.983 & 0.670 \\
			2 & 2 & 3 & 0.943 & 0.985 & 0.994 & 0.901 & 0.969 & 0.866 \\
			2 & 3 & 1 & 0.864 & 0.951 & 0.975 & 0.777 & 0.902 & 0.715 \\
			2 & 3 & 2 & 0.858 & 0.951 & 0.975 & 0.741 & 0.903 & 0.515 \\
			2 & 3 & 3 & 0.863 & 0.950 & 0.974 & 0.775 & 0.901 & 0.713 \\
			3 & 3 & 1 & 0.960 & 0.991 & 0.997 & 0.929 & 0.981 & 0.904 \\
			3 & 3 & 2 & 0.960 & 0.991 & 0.997 & 0.929 & 0.981 & 0.904 \\
			3 & 3 & 3 & 0.959 & 0.991 & 0.996 & 0.927 & 0.980 & 0.902 \\
			\hline
		\end{tabular}
	\end{center}
\end{sidewaystable}

\section*{Acknowledgements}
This work was supported by the Medical Research Council [grant number MC\_UP\_1302/2 to M.J.G. and A.P.M.]; and the National Institute for Health Research Cambridge Biomedical Research Centre [MC\_UP\_1302/6 to J.M.S.W.].

\end{document}